\PassOptionsToPackage{final}{microtype}
\documentclass[1p,times, authoryear]{elsarticle}

\usepackage[utf8]{inputenc}
\usepackage[T1]{fontenc}

\usepackage{amsmath,amsthm}
\usepackage{xspace}
\DeclareFontFamily{U}{mathb}{\hyphenchar\font45}
\DeclareFontShape{U}{mathb}{m}{n}{ <-6> mathb5 <6-7> mathb6 <7-8>
  mathb7 <8-9> mathb8 <9-10> mathb9 <10-12> mathb10 <12-> mathb12 }{}
\DeclareSymbolFont{mathb}{U}{mathb}{m}{n}
\DeclareMathSymbol{\prec}{\mathrel}{mathb}{"A0}
\DeclareMathSymbol{\succ}{\mathrel}{mathb}{"A1}
\DeclareMathSymbol{\preceq}{\mathrel}{mathb}{"A8}
\DeclareMathSymbol{\succeq}{\mathrel}{mathb}{"A9}
\DeclareMathSymbol{\precneq}{\mathrel}{mathb}{"AC}
\DeclareMathSymbol{\succneq}{\mathrel}{mathb}{"AD}

\setcitestyle{authoryear,open={(},close={)}}

\usepackage{color}
\definecolor{gray}{gray}{0.4}

\usepackage{amscd}
\usepackage{enumerate}
\usepackage[algoruled,vlined,english,linesnumbered]{algorithm2e}
\SetArgSty{textup}
\SetKwInOut{Input}{Input}\SetKwInOut{Output}{Output}
\SetKwIF{If}{ElseIf}{Else}{if}{:}{elif}{else:}{}%
\SetKw{KwTo}{in}\SetKwFor{For}{for}{\string:}{}%
\SetKwFor{While}{while}{:}{fintq}%
\AlgoDontDisplayBlockMarkers
\SetAlgoNoEnd
\makeatletter
\newcommand{\removelatexerror}{\let\@latex@error\@gobble}
\makeatother

\usepackage{bm} 

\usepackage[font=small]{caption}
\usepackage{subcaption}

\usepackage[english]{babel}
\usepackage[shortlabels,inline]{enumitem} 

\usepackage{tikz}
\usetikzlibrary{matrix}

\usepackage{fixme}
\fxsetup{author={T},envlayout={color},%
  layout={footnote,marginclue},theme=color}

\FXRegisterAuthor{ch}{ach}{C}

\usepackage{lipsum}

\usepackage[mathscr]{euscript}

\usepackage{array}

\usepackage{mathtools}
\mathtoolsset{showonlyrefs,showmanualtags}

\usepackage{tabu}
\usepackage{multirow}
\usepackage{booktabs}

\usepackage{siunitx}

\ifdefined\qty\else
  \ifdefined\NewCommandCopy
    \NewCommandCopy\qty\SI
  \else
    \NewDocumentCommand\qty{O{}mm}{\SI[#1]{#2}{#3}}
  \fi
\fi
\ifdefined\unit\else
  \ifdefined\NewCommandCopy
    \NewCommandCopy\unit\si
  \else
    \NewDocumentCommand\unit{O{}m}{\si[#1]{#2}}
  \fi
\fi

\usepackage[colorlinks,final]{hyperref}

\newtheorem{theorem}{Theorem}[section]

\newtheorem{proposition}[theorem]{Proposition}

\theoremstyle{definition}
\newtheorem{definition}[theorem]{Definition}

\theoremstyle{remark}
\newtheorem{example}[theorem]{Example}  
\newtheorem{remark}[theorem]{Remark}

\newcommand{\LM}{\mathrm{lm}}

\newcommand{\ZZ}{\mathbb{Z}}

\newcommand{\NN}{\mathbb{N}}

\newcommand{\QQ}{\mathbb{Q}}
\newcommand{\RR}{\mathbb{R}}

\newcommand{\HS}{\mathsf{HS}}

\newcommand{\Ker}{\mathrm{ker}}

\newcommand{\Span}{\mathrm{Span}}

\newcommand{\lm}{\LM}

\newcommand{\SPol}{\textup{S-Pol}}

\newcommand{\multiwhomo}{$\WW$-homogeneous\xspace}
\newcommand{\multiwdeg}{$\WW$-degree\xspace}

\newcommand{\dmax}{d_{\mathrm{max}}}
\newcommand{\WW}{\mathbf{W}}
\newcommand{\mdeg}{\mathrm{mdeg}}
\newcommand{\pdeg}{\Mdeg}
\newcommand{\Mdeg}{\mathrm{Mdeg}}
\newcommand{\wdeg}{\mathrm{wdeg}}
\newcommand{\Wnb}[1][]{\mathbf{W}_{\mathbf{n}}^{\rm b#1}}
\newcommand{\Wpb}[1][]{\mathbf{W}_{\mathbf{p}}^{\rm b#1}}
\newcommand{\WMgrevlex}[1][\mathbf{W}]{#1\textrm{-}\mathrm{Mdrl}}

\newcommand{\Mon}{\mathrm{Mon}}

\usepackage{microtype}

\newcommand{\pbox}[2][l]{%
  \begin{tabular}[c]{@{}#1@{}}#2\end{tabular}%
}

\begin{document}

\begin{frontmatter}
  \title{On the computation of Gröbner bases for matrix-weighted homogeneous systems}


  \author{Thibaut Verron}
  \ead{thibaut.verron@jku.at}

  \address{Institute for Algebra, Johannes Kepler University Linz, Austria}





  \begin{abstract}
    In this paper, we examine the structure of systems that are weighted homogeneous for several systems of weights, and how it impacts the computation of Gröbner bases. We present several linear algebra algorithms for computing Gröbner bases for systems with this structure, either directly or by reducing to existing structures. We also present suitable optimization techniques.


    As an opening towards complexity studies, we discuss potential definitions of regularity and prove that they are generic if non-empty. 
    Finally, we present experimental data from a prototype implementation of the algorithms in SageMath.
  \end{abstract}
\end{frontmatter}



\section{Introduction}
\label{sec:introduction}

Gröbner bases are a very powerful tool for manipulating polynomial ideals.
They give ways to describe the set of solutions of a system of polynomial equations, as well as obtaining various information of algebraic nature on the variety or the ideal.
Since their introduction by Buchberger \citep{Buchberger-1965}, many algorithms have been developed for computing Gröbner bases, either as adaptations of Buchberger's algorithm \citep{Giovini-1991, Gao-2015-new-framework-for, Faugere-2014-sparse} or by encoding polynomial operations in linear algebra \citep{Lazard-1983,Faugere-1999-F4, Faugere-2002-F5, Faugere-1993-fglm, Faugere-2014-subcubic-FGLM}.
A particularity of the computation of Gröbner bases is that it is proved to be an extremely difficult problem in the worst case, yet this does not materialize in practice, for systems arising in applications.
To better understand this phenomenon, research efforts have focused on identifying generic regularity properties of the polynomial systems ensuring better complexity bounds \citep{Giusti-1989-CastelnuovoRegularityCurves, Giusti-1993-DeterminationPointsIsoles,Bardet-2015}.
Those efforts were then extended to so-called structured polynomial systems, leading to dedicated algorithms with good complexity bounds.

Among such structures are multihomogeneous systems \citep{Caboara-1996, Faugere-2011-bihomo, Spaenlehauer-2012} and weighted homogeneous systems \citep{Faugere-2013-whomo-1, Faugere-2016-whomo-2}.
Both those structures have their particularities.
A Gröbner basis for multihomogeneous systems can be computed using a dedicated algorithm, grouping the computations according to the multidegree of the polynomials.
Characterizing the regularity of those systems, and analyzing the complexity of the algorithm, is however very complicated, and only done in a few cases such as bilinear systems.
For weighted homogeneous systems, one can perform a change of variables and use state-of-the-art algorithms, under standard regularity assumptions.
The additional structure leads to improved complexity bounds and performances.
In \cite{Bessonov-2022-ObtainingWeightsFor}, it was shown that algorithms for weighted homogeneous systems can also offer better performances for computing Gröbner bases of ideals where the chosen weights \emph{break} the homogeneity.

More recently, algorithms have been developed for computing sparse Gröbner bases \citep{Faugere-2014-sparse, Bender-2019-grobner-basis-over}, including the weighted and multihomogeneous cases.

In this work, we consider a structure which generalizes both the multihomogeneous and weighted homogeneous structures, by considering systems that are weighted homogeneous for several systems of weights~\citep{Kreuzer-2005-ComputationalCommutativeAlgebra}.
Equivalently, such systems are homogeneous for a graduation defined by a matrix of weights.

This structure can appear for instance in physics, where dimensional homogeneity is intrinsically multidimensional.
Consider for example the following vector equation
\begin{equation}
  \label{eq:39}
  m \mathbf{a} - q(\mathbf{E} + \mathbf{v} \times \mathbf{B}) = 0
\end{equation}
characterizing the movement (speed $\mathbf{v}$ in \unit{m.s^{-1}}, acceleration $\mathbf{a}$ in \unit{m.s^{-2}}) of a particle with mass $m$ (in \unit{kg}) and charge $q$ (in $\unit{As}$) under the action of an electric field $\mathbf{E}$ (in \unit{kg.m.A^{-1}.s^{-3}})  and a magnetic field $\mathbf{B}$ (in \unit{kg.A^{-1}.s^{-2}}).
This equation is dimensionally homogeneous, which means that all terms have the same degree in \unit{kg}, in \unit{m}, in \unit{s} and in \unit{A}.
In the language of weighted degrees, it means that the polynomial is weighted homogeneous for the $4$ systems of weights
\begin{equation}
  \label{eq:54}
  \begin{array}{r@{\,}cccccc@{\;}l}
      &\mathbf{v} & \mathbf{a} & m & \hspace{-0.3em}q & \mathbf{E} & \hspace{-0.1em}\mathbf{B} &\\
    W_{\unit{kg}} = \Big(& 0&0&1&0&1&1& \Big) \\
    W_{\unit{m}} = \Big(& 1&1&0&0&1&0& \Big) \\
    W_{\unit{s}} = \Big(& -2&-1&0&1&-3&-2& \Big) \\
    W_{\unit{A}} = \Big(& 0&0&0&1&-1&-1& \Big).
  \end{array}
\end{equation}
This structure may also appear as a construction for monomial orderings, in addition to preexisting homogeneity \citep{Collart-1997-walk}.

In \cite{Kreuzer-2005-ComputationalCommutativeAlgebra}, the computation of Gröbner bases for such systems is examined in a very general fashion, by refining the selection strategy in Buchberger's algorithm to compute a (matrix-weighted) homogeneous basis.
The algorithms which we present are adapted from the more recent algorithms based on linear algebra, and in particular the Matrix-F5 algorithm \citep{Bardet-2004}.
The adaptation of this algorithm to take advantage of the structure of the systems is a general approach that was key to the development of algorithms for the already mentioned homogeneous, multihomogeneous or weighted homogeneous structures, as well as the sparse systems.

Beyond the performance advantage given by fast linear algebra, the use of matrices allows for a clearer overview of the computations done in the algorithm, and it makes it possible to design further optimizations.
Specifically, in our case, the parallelization strategy used for multihomogeneous systems can be generalized (Proposition~\ref{prop:parallel}), and a new signature criterion can be used to eliminate entire matrices at once (Proposition \ref{prop:sig-gcd}).

To streamline the presentation, we introduce a general version of the Matrix-F5 algorithm, observing that for all those known structures, the specialized algorithm can be obtained by precomputing a list of computation steps and then running the Matrix-F5 algorithm with those steps.
The three algorithms which we present for matrix-weighted homogeneous systems also fall in this framework.

The first algorithm is dedicated to the matrix-weighted homogeneous structure and considers the polynomials matrix-weighted degree by matrix-weighted degree.
We then show that several optimizations are available for this algorithm, notably parallelization, and pruning some steps using signatures.

The second algorithm uses a change of variables to reduce the problem to the known multihomogeneous case.
The resulting algorithm performs exactly the same steps as the dedicated algorithm, but such a transformation makes it possible to reduce the implementation effort and to transfer existing optimizations.
This is similar in spirit to the weighted homogeneous case, which could be reduced via a change of variable to the homogeneous case.

Finally, we present the specialization of the sparse-Matrix-F5 \citep{Faugere-2014-sparse} algorithm to our structure, taking advantage of the sparsity without a particular treatment to the structure.
Again, this makes it possible to take advantage of existing algorithms, but with the caveat that the computation is restricted to the zero-dimensional case.

A general fact about Gröbner basis algorithms dedicated to a structure is that they compute a Gröbner basis for an order carefully chosen to ensure that the structure is preserved throughout the algorithm.
In addition to those theoretical properties which make it possible to obtain complexity bounds, this order usually gives the fastest Gröbner basis computation in practice.
For some applications, having any Gröbner basis is enough, for instance for computing the Hilbert series of the ideal \citep{Traverso-1996-hilbert-functions-buchberger, Bao-2021-ChiralRingsFutakia}.
Furthermore, multi-step strategies have been designed around this feature, combining those ``easy'' Gröbner basis computations with change of order algorithms to obtain a Gröbner basis for any wanted order.
In the multihomogeneous case, this ``easy'' ordering is a graded block ordering, and in the weighted homogeneous case, it is a weight-graded ordering.

In the matrix-weighted homogeneous case, similarly to those two cases, the most natural orders, and the one the aforementioned algorithms are tailored to, are orders graded with respect to the systems of weights.
It is a classical fact that any monomial ordering can be represented as such a graded ordering with respect to a matrix of weights.
This opens the possibility of looking at the situation from the opposite angle, by choosing the weights which represent the desired monomial ordering.
Of course, it is unlikely that the systems of interest are matrix-weighted homogeneous for those weights, but it is possible to adapt the algorithm to run on affine systems, i.e., not necessarily matrix-weighted homogeneous.
This opens different questions from the usual cases because it is unlikely that we can or want to avoid degree falls in this case.

The complexity of Gröbner basis algorithms is usually examined under some regularity assumptions which are proved to be satisfied for generic systems.
For the matrix-weighted homogeneous structure, similarly to the multihomogeneous case, it is not easy to define what a regular sequence is.
We examine several regularity properties, by analogy with known cases.
We show that those properties are generic if not empty.
This restriction is necessary and not unexpected, as this phenomenon was already present for the more specific weighted homogeneous structure.
With this caveat, the genericity of the properties is also not unexpected, as it stems from the construction and reduction of the matrices in the algorithm.
As an aside, with the general presentation of the algorithm, the proof of genericity immediately generalizes to other structures using a similar algorithmic template.
Finally, as an opening towards complexity analyses, we show how to compute the Hilbert series of ideals under some of those hypotheses.

We present some experimental results obtained with a prototype implementation of the algorithms\footnote{\url{https://gitlab.com/thibaut.verron/matrix-homo-gb}}, showing that the matrix-weighted strategy allows us to significantly reduce the size of the matrices and the time spent in reductions.

\paragraph{Structure of the paper}
In Section~\ref{sec:definitions}, we recall the main definitions concerning $\ZZ^{k}$-gradings on polynomial algebras, and we include properties that are relevant to the computation of Gröbner bases.
In Section~\ref{sec:general-matrix-f5}, we present algorithms for computing Gröbner bases for matrix-weighted homogeneous systems, starting with the general form of Matrix-F5 for structured systems.
We conclude that section with a discussion on inhomogeneous systems.
In Section~\ref{sec:hilb-mult-regul}, we discuss possible definitions of regularity, prove their conditional genericity, and show some of their early consequences.
Finally, in Section~\ref{sec:applications}, we report on experimental result.

\paragraph{Acknowledgements}
The author thanks the anonymous reviewers for their careful comments and suggestions, and in particular the reviewer who asked the insightful question prompting the discussion in Example~\ref{ex:semi-reg-hs}.

\paragraph{Funding}
This work was funded by the Austrian FWF, project P34872.

\section{Notations}
\label{sec:notations}

The following notations are used throughout the paper: $K$ is a field, $n \in \NN_{>0}$, and $A=K[X_{1},\dots,X_{n}] = K[\mathbf{X}]$.
A monomial of $A$ is an element of the form $X_{1}^{\alpha_{1}}\cdots X_{n}^{\alpha_{n}}$. We denote by $\Mon(A)$ the set of all monomials of $A$.
When the indexing is clear by the context, we shall frequently use bold shorthands for indexed tuples and products, such as $\mathbf{d} = (d_{1},\dots,d_{k})$ or $\mathbf{X}^{\bm{\alpha}} = X_{1}^{\alpha_{1}}\cdots X_{n}^{\alpha_{n}}$.

\section{Definitions}
\label{sec:definitions}

\subsection{Matrix gradings}
\label{sec:pluriw-homog}

We first recall the definitions of weighted homogeneity and multihomogeneity.

\begin{definition}
  A \emph{system of weights} on $A$ is a tuple $W{=}(w_{1},\dots,w_{n})$ in $\ZZ^{n}$.
  It defines a $\ZZ$-grading on $A$, by setting the \emph{weighted degree} (or $W$-degree) of a monomial to be
  \begin{equation}
    \label{eq:5}
    \deg_{W}(X_{1}^{\alpha_{1}}\cdots X_{n}^{\alpha_{n}}) = w_{1}\alpha_{1} + \dots + w_{n}\alpha_{n},
  \end{equation}
  and setting the weighted degree of a polynomial as the maximum of the weighted degrees of monomials in its support.
  A polynomial is called \emph{weighted homogeneous} or \emph{$W$-homogeneous} if all its monomials have the same $W$-degree.

  Let $k \in \{1,\dots,n\}$, and $\mathbf{n} = (n_{1},\dots,n_{k}) \in \NN^{k}$ such that $n=n_{1} + n_2 + \dots + n_{k}$.
  Group the variables $X_{1},\dots,X_{n}$ accordingly, as $X_{1,1},\dots,X_{1,n_{1}}, X_{2,1},\dots,X_{2,n_{2}},\dots,X_{k,n_{k}}$.
  The \emph{multidegree} of a monomial with respect to the partition $\mathbf{n}$ is
\begin{equation}
    \label{eq:17}
    \mdeg_{\mathbf{n}}(\mathbf{X}_{1,\bullet}^{\bm{\alpha_{1}}}\cdots X_{k,\bullet}^{\bm{\alpha}_{k}})
    = (\deg(\mathbf{X}_{1,\bullet}^{\bm{\alpha}_{1}}), \dots, \deg(\mathbf{X}_{k,\bullet}^{\bm{\alpha}_{k}})).
  \end{equation}
  A polynomial is called multihomogeneous (with respect to the partition $\mathbf{n}$) with multidegree $\mathbf{d}$ if all its monomials have multidegree $\mathbf{d}$.
  This defines a $\NN^{k}$-grading on $A$.
\end{definition}

In this paper, we consider a generalization of those two notions, grading polynomials according to several systems of weights.
We recall the main definitions and properties from~\cite[Sec.~4.1.A]{Kreuzer-2005-ComputationalCommutativeAlgebra}.

\begin{definition}[{\citet[Def.~4.1.6]{Kreuzer-2005-ComputationalCommutativeAlgebra}}]
   Let $k \in \{1,\dots,n\}$.
   A \emph{matrix of weights} is a matrix $\mathbf{W}=(w_{i,j}) \in \ZZ^{k \times n}$ with rank $k$.
   We denote by $W_{1},\dots,W_{k}$ its rows.
   The matrix-weighted degree (or  $\mathbf{W}$-degree) of the monomial $\mathbf{X}^{\bm{\alpha}}$ is
   \begin{equation}
     \label{eq:4}
     \pdeg_{\mathbf{W}}(\mathbf{X}^{\bm{\alpha}}) = \mathbf{W} \cdot \bm{\alpha} = (\deg_{W_{1}}(\mathbf{X}^{\bm{\alpha}}), \dots, \deg_{W_{k}}(\mathbf{X}^{\bm{\alpha}})).
   \end{equation}
   A polynomial $f$ is called $\mathbf{W}$-homogeneous (or matrix-weighted homogeneous) with matrix-weighted degree $\mathbf{d}$ if all the monomials in its support have matrix-weighted degree $\mathbf{d}$.
 \end{definition}   

 \begin{remark}
   In~\cite{Kreuzer-2005-ComputationalCommutativeAlgebra}, matrix-weighted homogeneity is seen as a general case of homogeneity, which allows the authors to simply refer to the corresponding notion as ``degree'', ``homogeneity'', etc., with the usual definition of homogeneity qualified as \emph{standard homogeneity}.
   In this work, we wish to preserve the distinction between (standard) homogeneity, weighted homogeneity, multihomogeneity, and matrix-weighted homogeneity, so we keep the naming specific.
 \end{remark}
 
\begin{example}
  Weighted homogeneity is a particular case of matrix-weighted homogeneity, by setting $k=1$ and $W_{1}=W$.
  
  Multihomogeneity is also a particular case, using as weights the block matrix of weights defined as follows. 
\end{example}

\begin{definition}
  Let $n \in \NN$, $k \in \{1,\dots,n\}$, and $\mathbf{n} = (n_1,\dots,n_k)$ a partition of $n$.
  The block matrix of weights $\Wnb=(W^{b}_{\mathbf{n},1},\dots,W^{b}_{\mathbf{n},k})$ is defined as
  \begin{equation}
    \label{eq:19}
    W^{b}_{\mathbf{n},i} = (\underbrace{0,\dots\dots,0}_{n_{1}+\dots+n_{i-1}},
    \underbrace{1,\dots,1}_{n_{i}},
    \underbrace{0,\dots\dots,0}_{n_{i+1}+\dots+n_{k}}).
  \end{equation}
\end{definition}

\begin{example}
  Consider the matrix of weights $\mathbf{W} =
  \begin{psmallmatrix}
    1 & 1 & 1 \\ 1 & 2 & 3
  \end{psmallmatrix}
  $ and the polynomial $f = X_{1}X_{2}^{2} + X_{1}^{2}X_3$.
  It is $\mathbf{W}$-homogeneous with $\mathbf{W}$-degree $(3,5)$.
\end{example}

\begin{remark}
  Just like in the multihomogeneous case, there is no canonical way to define the matrix-weighted degree of an arbitrary polynomial.
\end{remark}

It is a straightforward verification that this defines a $\ZZ^{k}$-grading on $A$.
The $\mathbf{W}$-homogeneous components of a polynomial, and $\mathbf{W}$-homogeneous ideals, are defined as usual for this graduation.
For a $\mathbf{W}$-homogeneous ideal $I$ (possibly 0), we denote by $(A/I)_{\mathbf{d}}$ the $K$-vector span of all $\mathbf{W}$-homogeneous polynomials of $\mathbf{W}$-degree $\mathbf{d}$ in $A/I$.

\begin{definition}
    Two matrices of weights $\mathbf{W}, \mathbf{W}'$ are called \emph{equivalent} if there exists an invertible matrix $P \in \QQ^{k \times k}$ such that
  \begin{equation}
    \label{eq:8}
    \mathbf{W} = P \cdot \mathbf{W}' .
  \end{equation}
\end{definition}

Equivalent matrices of weights are exactly those which have the same matrix-weighted homogeneous components.

\begin{proposition}
  \label{prop:charac-equiv-w}
  Let $\mathbf{W}_{1}, \mathbf{W}_{2}$ be two matrices of weights.
  Then the following are equivalent:
  \begin{enumerate}
    \item $\mathbf{W}_{1}$ and $\mathbf{W}_{2}$ are equivalent
    \item for all $m_{1}, m_{2}$ monomials in $K[\mathbf{X}]$,
    \begin{equation}
      \label{eq:9}
      \Mdeg_{\mathbf{W}_{1}}(m_1) = \Mdeg_{\mathbf{W}_{1}}(m_2) \iff \Mdeg_{\mathbf{W}_{2}}(m_1) = \Mdeg_{\mathbf{W}_{2}}(m_2).
    \end{equation}
    \item for all $f \in K[\mathbf{X}]$, $f$ is $\mathbf{W}_{1}$-homogeneous iff $f$ is $\mathbf{W}_{2}$-homogeneous.
  \end{enumerate}\end{proposition}
\begin{proof}
  Let $\alpha^{(1)}$ and $\alpha^{(2)}$ be the exponent vectors of $m_{1}$ and $m_{2}$ respectively.
  If $\mathbf{W}_{1}$ and $\mathbf{W}_{2}$ are equivalent, then there exists $P$ invertible such that $\mathbf{W}_{1} = P\mathbf{W}_{2}$.
  So for $i \in \{1,2\}$, $\deg_{W_{1}}(m_{i}) =\mathbf{W}_{1}\cdot  \alpha^{(i)}  = P \cdot \mathbf{W}_{2} \cdot \alpha^{(i)}= P \cdot \deg_{W_{1}}(m_{i})$.
  Property~\eqref{eq:9} follows.

  Conversely, assume that Property 2 holds. 
  Let $v \in \QQ^{n}$ be a vector in $\ker(\mathbf{W}_{1})$.
  Clearing out the denominators, we can assume that $v \in \ZZ^{n}$.
  Separating the non-negative and non-positive coordinates, we can write $v = \alpha - \beta$ with $\alpha, \beta \in \NN^{n}$.
  Then $\mathbf{W}_{1}\alpha  = \mathbf{W}_{1}\beta$.
  Taking the corresponding monomials, using Property~2 and tracing back the previous computation, we get that $v \in \ker(\mathbf{W}_{2})$.
  By symmetry, $\ker(\mathbf{W}_{1})=\ker(\mathbf{W}_{2})$.
  This implies that $\mathbf{W}_{1}$ and $\mathbf{W}_{2}$ have the same row space, and thus the same reduced row echelon form.
  So they are equal up to left-multiplication by an invertible matrix.

  Properties 2 and 3 are clearly equivalent.
\end{proof}

\subsection{Size-bounded matrix of weights}
\label{sec:posit-size-bound}

So far we did not put any restrictions on the weights.
In particular, they may be zero or negative and lead to matrix-weighted degrees which are zero or negative.
In this paper, we will restrict to matrices of weights such that the number of monomials at any \multiwdeg is finite.
We recall the corresponding definitions and properties from~\cite[Sec.~4.1.C]{Kreuzer-2005-ComputationalCommutativeAlgebra}.

\begin{definition}
  A matrix of weights $\mathbf{W} = (w_{i,j}) \in \ZZ^{n \times k}$ is:
  \begin{itemize}
    \item \emph{positive} (resp. \emph{non-negative}) if all the entries $w_{i,j}$ are positive (resp. \emph{non-negative});
    \item \emph{of positive type} \citep[Def.~4.1.17]{Kreuzer-2005-ComputationalCommutativeAlgebra} if there exists $a_1, \dots, a_n \in \ZZ$ such that $a_1W_1 + \dots + a_n W_n$ is positive;
    \item \emph{size-bounded} if for all $\mathbf{d} \in \ZZ^{k}$, there exists only finitely many monomials with \multiwdeg $\mathbf{d}$.
  \end{itemize}
\end{definition}

The notion of positive matrices of weights is not more restrictive than that of positive type:
\begin{proposition}
  Let $\mathbf{W}$ be a matrix of weights, the following are equivalent:
  \begin{enumerate}
    \item $\mathbf{W}$ is of positive-type;
    \item $\mathbf{W}$ is equivalent to a matrix of weights with a positive row;
    \item $\mathbf{W}$ is equivalent to a positive matrix of weights
  \end{enumerate}
\end{proposition}
\begin{proof}
  Clearly (3) implies (1).
  If $\mathbf{W}$ is of positive type, there exists $a_1, \dots, a_n \in \ZZ$ such that $a_1W_1 + \dots + a_n W_n$ is positive.
  Let $i$ be such that $a_i$ is nonzero, and let
   $P$ be the transformation matrix replacing the $i$'th row by $\sum_{i} a_i W_i$.
  Clearly $P\mathbf{W}$ is equivalent to $\mathbf{W}$ and its $i$'th row is positive.

  If one of $\mathbf{W}$'s rows is positive, adding sufficiently many copies of that row to the other rows, one obtains an equivalent matrix whose entries are all positive.
\end{proof}

\begin{remark}
  The notion of \emph{positive matrix} defined in~\cite[Def.~4.2.4]{Kreuzer-2005-ComputationalCommutativeAlgebra} is different.
\end{remark}

Therefore, matrices of weights of positive type are size-bounded (see also~\cite[Prop.~4.1.19]{Kreuzer-2005-ComputationalCommutativeAlgebra}).
More generally, size-bounded matrices of weights can be characterized by the sign of coordinates of vectors in their right kernel.
\begin{proposition}
  \label{prop:size-bounded}
  Let $\mathbf{W}$ be a matrix of weights.
  Then $\mathbf{W}$ is size-bounded if and only if there does not exist any non-trivial monomial with \multiwdeg $(0,\dots,0)$.
\end{proposition}
\begin{proof}
  Assume that there exists a non-trivial monomial $m_0$ with \multiwdeg $(0,\dots,0)$.
  Then $m_0^{k}$ has \multiwdeg $(0, \dots, 0)$ for all $k \in \NN$, and $\mathbf{W}$ is not size-bounded.

  Conversely, assume that $\mathbf{W}$ is not size-bounded, and let $\mathbf{d}$ be a \multiwdeg such that the set $M_{\mathbf{d}}$ of monomials of \multiwdeg $\mathbf{d}$ is infinite.
  By Dickson's lemma, there exists a finite subset $B_{\mathbf{d}} \subseteq M_{\mathbf{d}}$ such that all elements of $M_{\mathbf{d}}$ are divisible by an element of $B_{\mathbf{d}}$.
  Since $M_{\mathbf{d}}$ is infinite, there exists a monomial $m \in M_{\mathbf{d}} \setminus B_{\mathbf{d}}$, and by construction it is divisible by $b \in B_{\mathbf{d}}$.
  The quotient $m/b$ is non-trivial and has \multiwdeg $\mathbf{d}-\mathbf{d} = (0,\dots,0)$.
\end{proof}

\begin{remark}
  In particular, the two consequences of~\cite[Prop.~4.1.19]{Kreuzer-2005-ComputationalCommutativeAlgebra} are equivalent for free finitely-generated $A$-modules.
\end{remark}

In the rest of the paper, we will frequently restrict to matrices of weights such that $W_{1}$ is positive.

\subsection{Matrix-weighted degree-reverse-lexicographical ordering}
\label{sec:matrix-gr-reverse-lexic}
By analogy with the weighted and multihomogeneous cases, we define the following order, which will be a natural order for computing Gröbner bases.
Recall that the \emph{reverse-lexicographic ordering} is defined by
\begin{equation}
  \label{eq:6}
  \mathbf{X}^{\alpha} <_{\mathrm{revlex}} \mathbf{X}^{\beta} \iff
  \begin{cases}
    \alpha_{n} = \beta_{n},  \dots,
    \alpha_{j+1} = \beta_{j+1}\\
    \alpha_{j} > \beta_{j}.
  \end{cases}
\end{equation}

\begin{definition}
  Let $\mathbf{W} \in \ZZ^{k \times n}$ be a matrix of weights.
  The \emph{$\mathbf{W}$-matrix-weighted degree reverse lexicographical ordering} on $A$ is the order $<_{\mathbf{W}-\mathrm{Mdrl}}$, defined by 
  \begin{align}
    \label{eq:21}
    m <_{\mathbf{W}-\mathrm{Mdrl}} m'  
   \iff
      \begin{cases}
        \pdeg_{\mathbf{W}}(m) <_{\mathrm{lex}} \pdeg_{\mathbf{W}}(m') \text{, or }\\
        \pdeg_{\mathbf{W}}(m) = \pdeg_{\mathbf{W}}(m')
        \text{ and } m <_{\mathrm{revlex}} m'.
    \end{cases}
  \end{align}
\end{definition}

\begin{remark}
  The algorithm from Section~\ref{sec:general-matrix-f5} would also work with any other tie-breaker order instead of revlex. The revlex ordering is commonly considered to lead to the most efficient Gröbner basis computations in classical settings, and we stick to it for the remainder of the paper.
\end{remark}

\begin{proposition}
  Let $\mathbf{W}=(W_{1},\dots,W_{k})$ be a matrix of weights.
  Then the $\mathbf{W}$-Mdrl ordering is a total ordering, and it is compatible with the $\ZZ^{k}$-graduation induced by $\mathbf{W}$ (lexicographically ordered).
  
  Further assume that either of the following conditions holds:
  \begin{enumerate}
    \item $\mathbf{W}$ is non-negative and size-bounded;
    \item $W_{1}$ is positive.
  \end{enumerate}
  Then the $\mathbf{W}$-Mdrl ordering is a monomial ordering (i.e., for any nontrivial monomial $m$, $1 <_{\WMgrevlex} m$).
\end{proposition}
\begin{proof}
  Clearly, the $\mathbf{W}$-Mdrl ordering is a total ordering and compatible with the graduation.

  Now let $m$ be a nontrivial monomial with matrix-weighted degree $\mathbf{d}$, and assume that condition (1) holds.
  Since $\mathbf{W}$ is size-bounded, by Proposition~\ref{prop:size-bounded}, $\mathbf{d}$ is not zero.
  Furthermore, since $\mathbf{W}$ is non-negative, the non-zero coefficients of $\mathbf{d}$ are positive.
  In particular, $0 <_{\mathrm{lex}} \mathbf{d}$, and thus $1 <_{\WMgrevlex} m$.

  If instead condition (2) holds, then the first coefficient of $\mathbf{d}$ is positive, with the same conclusion.
\end{proof}

\begin{remark}
  \label{rmk:equiv-weights-orders}
  Let $\mathbf{W}$ and $\mathbf{W}'$ be two equivalent matrices of weights.
  Let $m$ and $m'$ be two monomials of $A$.
  It is not true in general that $m <_{\WMgrevlex} m'$ if and only if $m <_{\WMgrevlex[\mathbf{W}']} m'$.

  With $n=2$, take for example
  $\mathbf{W} =
    \begin{psmallmatrix}
      1 & 1 \\
      0 & 1 \\
    \end{psmallmatrix}$ and 
    $\mathbf{W}' =
    \begin{psmallmatrix}
      -1 & 0 \\
      0 & 1 \\
    \end{psmallmatrix}
    \cdot \mathbf{W}
    = \begin{psmallmatrix}
      -1 & -1 \\
      0 & 1 \\
    \end{psmallmatrix}$.
  Then $1 <_{\WMgrevlex} XY$ but $XY <_{\WMgrevlex[\mathbf{W}']} 1$.
  In particular, the order $<_{\WMgrevlex[\mathbf{W}']}$ is not even a monomial ordering, contrary to $<_{\WMgrevlex[\mathbf{W}]}$.
\end{remark}

Despite the remark, since equivalent matrices of weights define the same homogeneous components,
Gröbner bases \emph{for matrix-weighted homogeneous systems} do not depend on the choice of an equivalent matrix of weights.
\begin{proposition}
  \label{prop:equiv-w-gb}
  Let $\mathbf{W}$ and $\mathbf{W}'$ be equivalent matrices of weights.
  Let $I$ be a $\mathbf{W}$-matrix-weighted homogeneous ideal.
  Let $L$ (resp. $L'$) be the set of leading monomials of $I$ w.r.t. $<_{\WMgrevlex}$ (resp. $<_{\WMgrevlex[\mathbf{W}']}$).
  Then $L=L'$.

  In particular, if $G \subseteq I$ is a Gröbner basis of $I$ w.r.t. $<_{\WMgrevlex}$ then $G$ is a Gröbner basis of $I$ w.r.t $<_{\WMgrevlex[\mathbf{W}']}$.
\end{proposition}
\begin{proof}
  Let $m \in L$, and let $f \in I$ such that $m = \lm(f)$ w.r.t. $<_{\WMgrevlex}$.
  Since $I$ is a $\mathbf{W}$-matrix-weighted homogeneous ideal, all matrix-weighted homogeneous components of $f$ lie in $I$, so w.l.o.g. we can assume that $f$ is matrix-weighted homogeneous.
  In particular, all the terms in its support have the same $\mathbf{W}$-matrix-weighted degree, so $m$ is also the leading monomial of $m$ w.r.t. the revlex order.
  And since $\mathbf{W}$ and $\mathbf{W}'$ are equivalent, $f$ is also $\mathbf{W}'$-matrix-weighted homogeneous (Proposition~\ref{prop:charac-equiv-w}), and by the same argument, $m$ is the leading monomial w.r.t. $<_{\WMgrevlex[\mathbf{W}']}$.
  So $m \in L'$.
  We conclude that $L \subseteq L'$, and by symmetry $L = L'$.

  The consequences on the Gröbner bases is immediate since Gröbner basis are characterized by their leading terms.
\end{proof}

\begin{remark}
  With the hypotheses of the proposition, one can observe that the Gröbner basis $G$ is in fact a Gröbner basis for the revlex ordering, and conversely, that any Gröbner basis w.r.t. revlex and comprised only of $\mathbf{W}$-homogeneous polynomials is a Gröbner basis w.r.t. $<_{\WMgrevlex}$ and $<_{\WMgrevlex[\mathbf{W}']}$.
  In that sense, the only observation of the proposition is that $I$ must be $\mathbf{W}'$-homogeneous.

  This is similar to the classical homogeneous case: for homogeneous systems, grevlex Gröbner bases are essentially revlex Gröbner bases.

  The use of graded orders matters when discussing inhomogeneous systems (see the next remark and Section~\ref{sec:affine-systems})
\end{remark}

\begin{remark}
  If $k=n$, there is only one monomial at each $\mathbf{W}$-matrix-weighted degree, and the $\WMgrevlex$ ordering is defined only in terms of the matrix-weighted degrees.
  More generally, it is a classical fact~\citep{Robbiano-1985-TermOrderingsPolynomial} that for any monomial ordering $<$ on $A$, there exists a matrix of weights $\mathbf{W} \in \RR^{n \times n}$ such that
  \begin{equation}
    \label{eq:50}
    m \,{<}\, m' \iff \pdeg_{\mathbf{W}}(m) <_{\mathrm{lex}} \pdeg_{\mathbf{W}}(m') \iff m <_{\WMgrevlex} m'.
\end{equation}
\end{remark}

\section{Matrix-F5 algorithm}
\label{sec:general-matrix-f5}

\subsection{General Matrix-F5 algorithm for structured systems}
\label{sec:generic-matrix-f5}

We will present different variants of the Matrix-F5 algorithm adapted to the structure of matrix-weighted systems.
To separate the handling of the structure from the presentation of the algorithm, we start by presenting a general form of the Matrix-F5 algorithm allowing us to take advantage of different structures.

\newcommand{\Mac}{\textup{\texttt{Mac}}}
\begin{algorithm}
  \Input{$F=(f_{1},\dots,f_{r})$ a system of structured polynomials; $\dmax$ a bound}
  \Output{$G$ a Gröbner basis of $F$, truncated at the bound $\dmax$}
  \texttt{Steps} $\leftarrow \textsf{AlgoSteps}(F,\dmax)$\;
  $G \leftarrow \{(f_{i},(1,i)) : i \in \{1,\dots,r\}\}$\;
  \For{\label{algoline:big-loop-beg}$(\textbf{d},M,S)$ in \texttt{Steps}}{
    $\Mac_{\mathbf{d}} \leftarrow$ matrix with columns indexed by $M$ and $0$ rows\;
    \For{$i \in \{1,\dots,r\}$}{
      $S_{i} \leftarrow \{m \text{ for $(m,j) \in S$ with $j=i$}\}$\;
      \For{$m \in S_{i}$ in increasing order}{
        \If{$(m,j)$ can be eliminated by a criterion (Section~\ref{sec:optimizations})\label{algoline:criteria}}
        {Pass}
        \Else{
        Find an element $g \in G$ with signature $(m',i)$ such that $m'$ divides $m$
        \label{algoline:prev-elt}\;
        Add the coefficient vector of $\frac{m}{m'}g$ to the bottom of $\Mac_{d}$, with signature $(m,i)$
        \label{algoline:new-row}
      }
      }
      Reduce $\Mac_{\mathbf{d}}$ to row-echelon form, only reducing below the pivots and without permutations
      \label{algoline:echelon}\;
      Add to $G$ all polynomials corresponding to a new pivot in $\Mac_{\mathbf{d}}$ and which are not reducible by $G$, with the signature of their row
      \label{algoline:new-gb}\;
    }
  }
  \Return{$G$}
  \caption{General Matrix-F5 algorithm for structured systems}
  \label{algo:matrixF5-phomo}
\end{algorithm}

For simplicity, we defer the presentation of the signature criteria at line~\ref{algoline:criteria}, and in particular of the F5 criterion, to Section~\ref{sec:optimizations}.

The algorithm as presented is very close to the original Matrix-F5 algorithm, but with an additional degree of freedom in the order for processing the new elements.
That difference is delegated to the subroutine \textsf{AlgoSteps} which will be instantiated to obtain different variants of the algorithm (for an example of instantiation, see Algorithm~\ref{algo:matrixF5-phomo} and the accompanying comments).
For now, we only give a specification for this routine:
it takes the same input as the main algorithm, and it returns a list of ``steps'' to process, each of which is a triple $(\mathbf{d},M,S)$ where:
\begin{itemize}
  \item $\mathbf{d}$ is some ``degree'' for the step;
  \item $M$ is the set of all monomials of ``degree'' $\mathbf{d}$  in $A$;
  \item $S$ is a set of pairs $(m,i) \in \Mon(A) \times \{1,\dots,m\}$ such that $mf_{i}$ has ``degree'' $\mathbf{d}$.
\end{itemize}
The notion of ``degree'' can for example be the degree given by a $\NN$-graduation on $A$ (standard degree, weighted degree), a $\NN^{k}$-graduation (multidegree, matrix-weighted degree), or even a graduation on a subalgebra of $A$ containing all the polynomials in $F$ (sparse degree~\citep{Faugere-2014-sparse}, $G$-degree~\citep{Faugere-2013-invariant}).
We assume that the ``degrees'' are totally ordered by an order $<$, and that it is a well-order.
This assumption holds for all the previous examples, except for the $G$-degree. For the $G$-degree, the authors made it into an ordering degree by first grading by total degree.

The pairs $(m,i)$ are called \emph{signatures}~\citep{Faugere-2002-F5}, and with the notations above, we say that the signature $(m,i)$ has ``degree'' $\mathbf{d}$.
Going in depth in the theory of signature Gröbner bases is out of the scope of this paper, so we only recall the basic requirements to prove that the algorithm and the criteria in Section~\ref{sec:optimizations} are correct.
The interested reader will find more details in~\cite{Faugere-2002-F5, Roune-2012-PracticalGrobnerBasis, Arri-2010-F5CriterionRevised, Gao-2015-new-framework-for, Eder-2017-survey}.

In this paper, signatures are ordered with the ``degree over position over term'' order: let $(m,i)$ and $(m,i')$ be two signatures with respective ``degree'' $\mathbf{d}$ and $\mathbf{d}'$, 
\begin{equation}
  \label{eq:43}
  (m,i) \prec (m',i') \iff
  \begin{cases}
    \mathbf{d} < \mathbf{d}' \\
    \text{or } \mathbf{d} = \mathbf{d}' \text{ and } i < i' \\
    \text{or } \mathbf{d} = \mathbf{d}' \text{ and } i = i' \text{ and } m < m'.
  \end{cases}
\end{equation}
Note that by assumption the ``degrees'' are totally well ordered, so the signature ordering is a total order and a well order.

For $f \in \langle f_{1},\dots,f_{r}\rangle$, we say that $f$ has signature $(m,i)$ if it can be written
\begin{equation}
  \label{eq:44}
  f = \alpha m f_{i} + \sum_{j\in \{1,\dots,i\}} \sum_{\mu \in \Mon(A)} \alpha_{\mu,j}\mu f_{j}
\end{equation}
with $\alpha, \alpha_{\mu,j} \in K$, $\alpha \neq 0$, and for all signature $(\mu,j)$ with $\alpha_{\mu,j} \neq 0$, $(\mu,j) \prec (m,i)$.

Let $(m,i)$ be a signature.
For all signatures $(m,i)$, we denote by $V_{(m,i)}$ (resp. $V_{\prec (m,i)}$) the $K$-vector space spanned by all polynomials which have a signature smaller than, equal or incomparable to $(m,i)$ (resp. strictly smaller than $(m,i)$).
Concretely, $V_{(m,i)} = \Span(\mu f_{j} : (\mu,j) \preceq (m,i))$ and $V_{\prec (m,i)} = \Span(\mu f_{j} : (\mu,j) \precneq (m,i))$.

\begin{remark}
  A given polynomial can have several signatures.
  In the algorithm however, polynomials are computed and stored in memory with a signature, which allows us to talk about \emph{the} signature of a polynomial.
  An alternative point of view is that since the signature ordering is a well order, a given polynomial in the ideal admits a minimal signature, and we shall prove in Section~\ref{sec:multi-mathbfw-homog} that in our case (as in all known cases) this is the one computed in the algorithm.
\end{remark}

The reason for restricting reductions in line~\ref{algoline:echelon} is to maintain this correspondence between signatures and rows in the echelon form.
An alternative is to modify the algorithm as follows: 
\begin{itemize}
  \item at line~\ref{algoline:echelon}, reduce to row echelon form without restrictions on the pivots
  \item at line~\ref{algoline:new-gb}, the polynomial does not necessarily match the signature of its row, so only store the polynomial in $G$, without any signature
  \item at line~\ref{algoline:prev-elt}, pick the only suitable polynomial for which a signature is known, that is $g=f_{i}$ with signature $(1,i)$; in other words, form the row corresponding to $mf_{i}$.
\end{itemize}
With that modified algorithm, we end up reducing $mf_{i}$ multiple times, but on the other hand, it allows the algorithm to use fast linear algebra for the computation of the row echelon form.

The correctness of both variants of the algorithm will be justified in Section~\ref{sec:multi-mathbfw-homog}.

\begin{remark}
  A key feature of the algorithm is that the list of steps to process is computed at the very start, and does not depend on polynomials computed throughout the algorithm.
  With the variant disregarding the signatures at each step, the matrices are built using only the input, and intermediate results only affect subsequent steps \emph{via} the criteria.

  This is not surprising if one thinks of the algorithm as a refined version of Lazard's algorithm which builds the full Macaulay matrix of the ideal and reduces it.
  The refinement in the Matrix-F5 algorithm is to use signatures to exclude some rows of the Macaulay matrix, and graduations to cut it into smaller independent matrices.
\end{remark}

\subsection{Matrix-weighted homogeneous Matrix-F5 algorithm}
\label{sec:multi-mathbfw-homog}

We now specialize Algorithm~\ref{algo:matrixF5-phomo} to weighted-matrix-weighted gradings.
Let $\mathbf{W}$ be a matrix of weights such that $W_{1}$ is positive.
Let $\dmax \in \NN$.
The algorithm will compute a Gröbner basis of $F$ w.r.t. the $\WMgrevlex$ ordering, truncated at $W_{1}$-degree $\dmax$.

This is done by letting the routine \textsf{AlgoSteps} output steps following the grading.
Thus, the ``degree'' $\mathbf{d}$ is chosen to be the $\mathbf{W}$-matrix-weighted degree, and we consider all such $\mathbf{W}$-degrees with $d_{1} \leq \dmax$.
For each step, corresponding to a $\mathbf{W}$-degree $\mathbf{d}$, the set of monomials and signatures is exactly those which realize the $\mathbf{W}$-degree $\mathbf{d}$.
Matrix-weighted degrees are ordered using the lexicographic order so that signatures are totally ordered. We will however see (Proposition~\ref{prop:parallel}) that the algorithm does not need to follow this order.

\begin{algorithm}
  \Input{$F=(f_1,\dots,f_{r})$ a system of $\mathbf{W}$-homogeneous polynomials with respective matrix-weighted degree $\mathbf{d}_{1},\dots,\mathbf{d}_{\mathbf{m}}$; $\dmax \in \NN$}
  \Output{$\texttt{Steps}$ a list of all steps to process to compute a Gröbner basis up to $W_{1}$-degree $\dmax$}
  $\texttt{Steps} \leftarrow []$\;
  \For{$d \in \{1,\dots,\dmax\}$}{
    $M \leftarrow$ set of monomials in $A$ with $W_{1}$-degree $d$\;
    $\texttt{Degs} \leftarrow \{\Mdeg_{\mathbf{W}}(m) : m \in M\}$\;
    \For{$\mathbf{d} \in \texttt{Degs}$ in increasing order}{
      $\texttt{Mons} \leftarrow 
      $ set of monomials of $\mathbf{W}$-degree $\mathbf{d}$ in $M$\;
      $\texttt{Sigs} \leftarrow \left\{ (m,i) : i \in \{1,\dots,r\}, \text{ $m$ monomial with } \Mdeg_{\mathbf{W}}(m) =\mathbf{d}-\mathbf{d}_{i} \right\}$\;
      Append $(\mathbf{d},\texttt{Mons},\texttt{Sigs})$ to $\texttt{Steps}$\;
    }
  }
  \Return{$\texttt{Steps}$}
  \caption{AlgoSteps for matrix-weighted homogeneous systems}
  \label{algo:algosteps-matrixhomo}
\end{algorithm}

\begin{remark}
  If $W_{1}$ is not positive but $\mathbf{W}$ is size-bounded, it is possible to design a suitable \textsf{AlgoSteps} function, by using a linear combination of the weights instead of $W_{1}$.
  Functionally, the algorithm is the same, for an equivalent system of weights.
  
  If $\mathbf{W}$ is not size-bounded, the algorithm as presented will not work, because the \textsf{AlgoSteps} subroutine will not terminate.
  However, algorithms such as F4 or F5, which build matrices at each step dynamically, based on the support of the polynomials being considered, would be able to follow a similar structure.
\end{remark}

We now prove the main loop invariants for the algorithm without criteria, which ensure that the algorithm is correct.
The proof is a transposition of the result in the homogeneous case~\citep{Bardet-2015}, and we only outline the arguments.
\begin{theorem}[Loop invariants]
  Let $\mathbf{d}$ be a matrix-weighted degree with $d_{1} \leq \dmax$, and $(m,i)$ be a signature of matrix-weighted degree $\mathbf{d}$ added in $\Mac_{\mathbf{d}}$.
  \begin{enumerate}
    \item the polynomial $\frac{m}{m'}g$ added at line~\ref{algoline:new-row} has signature $(m,i)$;
    \item the polynomials associated to the rows above and up to that signature, together with the polynomials associated to the rows in the previous matrices, span $V_{(m,i)}$;
    \item when starting step $\mathbf{d}$ (line~\ref{algoline:big-loop-beg}), any polynomial with signature at most $(m,i)$ with $(m,i)$ having matrix-weighted degree $< \mathbf{d}$ is reducible modulo $G$.
  \end{enumerate}
\end{theorem}
\begin{proof}
  For the first invariant, we proceed by induction on the signature $(m,i)$, the base case being clear.
  Let $h = \frac{m}{m'}g$ be the polynomial added to the matrix (line~\ref{algoline:new-row}).
  Since $g'$ was added with signature $(m',i)$, by induction hypothesis it has signature $(m',i)$, and therefore it can be written as
  \begin{equation}
    \label{eq:22}
    g' = \alpha m'f_{i} + \sum_{(\mu',j) \prec (m',i)} \alpha_{\mu',j} \mu' f_{j},
  \end{equation}
  with $\alpha, \alpha_{\mu,j} \in K$ and $\alpha \neq 0$.
  As a consequence,
  \begin{equation}
    \label{eq:30}
    h = \frac{m}{m'}g = \alpha m f_{i} + \sum_{(\mu',j) \prec (m',i)} \alpha_{\mu',j} \frac{m}{m'}\mu' f_{j}.
  \end{equation}
  The order on the signatures is compatible with the monomial product, so $(\mu',j) \prec (m',i)$ iff $(\frac{m}{m'}\mu',j) \prec (m,i)$. Therefore $h$ has signature $(m,i)$.

  For the second invariant, we proceed again by induction on the signature, the base case being clear.
  Assume that the induction hypothesis is proved for all rows processed before $(m,i)$. In particular, since the signatures are totally ordered, all rows processed so far span $V_{\prec (m,i)}$.
  We will prove by double inclusion that $\Span(V_{\prec (m,i)},h) = V_{(m,i)}$.
  Since $h$ has signature $(m,i)$, $\Span(V_{\prec (m,i)},h) \subseteq V_{(m,i)}$.

  For the reverse inclusion, write
  \begin{equation}
    \label{eq:15}
    h = \alpha mf_{i} + \sum_{(\mu,j) \prec (m,i)} \alpha_{\mu,j} \mu f_{j},
  \end{equation}
  with $\alpha, \alpha_{\mu,j} \in K$ and $\alpha \neq 0$.
  The second summand lies in $V_{\prec (m,j)}$, therefore $mf_{i} \in \Span(V_{\prec (m,i)},h)$.
  Therefore $\Span(V_{\prec (m,i)},h) = V_{(m,i)}$, and by induction, the second invariant is proved.

  The third invariant is a consequence of the characterization of Gröbner bases using Macaulay matrices: since we pick all the possible signatures in $S_{\mathbf{d}}$ for each $\mathbf{d}$, by the second invariant, at the end of the loop, the rows of $\Mac_{\mathbf{d}}$ span the space of polynomials of matrix-weighted degree $\mathbf{d}$, and they have distinct leading monomials.
  This implies that all the leading monomials of $\langle F \rangle$ with a $\mathbf{W}$-degree of $\mathbf{d}$ are a pivot in $\Mac_{\mathbf{d}}$.
  Since again this is done for every $\mathbf{d}$, all the leading monomials of $\langle F \rangle$ with a $\mathbf{W}$-degree of at most $\mathbf{d}$ are a pivot in one of the matrices computed so far.
  By construction of $G$, it implies that every polynomial with of a $\mathbf{W}$-degree at most $\mathbf{d}$ is reducible by $G$.
\end{proof}

\begin{remark}
  \label{rmk:choice-prev-step}
  In particular, the correctness of the algorithm does not depend on the choice of $g$ at line~\ref{algoline:prev-elt}.
\end{remark}

\subsection{Optimizations}
\label{sec:optimizations}

We present here a few optimizations that are available for the algorithm.
First, if $W_{1}$ is a positive system of weights, one can parallelize the computations like in the multihomogeneous case.

\begin{proposition}
  \label{prop:parallel}
  Let $\mathbf{W}$ be a matrix of weights such that $W_1$ is positive.
  Let $\mathbf{d}=(d_1,\dots,d_{k})$ and $\mathbf{d}'=(d'_{1},\dots,d'_{k})$ be two distinct matrix-weighted degrees, with $d_1=d'_{1}$.
  Let $M_{\mathbf{d}}$ and $S_{\mathbf{d}}$ (resp. $M_{\mathbf{d}'}$ and $S_{\mathbf{d}'}$) the set of monomials and signatures for the step $\mathbf{d}$ (resp. $\mathbf{d}'$).
  Then no monomial of $M_{\mathbf{d}}$ divides any monomial of $M_{\mathbf{d}'}$, and no signature of $S_{\mathbf{d}}$ divides any signature of $S_{\mathbf{d}'}$.
\end{proposition}
\begin{proof}
  For the sake of contradiction, let $\mu \in M_{\mathbf{d}}$ and $\mu' \in M_{\mathbf{d}'}$ be such that there exists  a monomial $\nu$ with $\mu' = \nu \mu$.
  Since the monomials have distinct matrix-weighted degree, $\mu \neq \mu'$, and therefore $\nu \neq 1$.
  Since $W_1$ is a positive system of weights, $\wdeg_{W_1}(\nu) > 0$.
  So $d'_{1} = \wdeg_{W_1}(\mu') = \wdeg_{W_1}(\nu) + \wdeg_{W_1}(\mu) > \wdeg_{W_1}(\mu) = d_1$, which contradicts the assumption.
  The proof is the same for the signatures.
\end{proof}

The proposition implies that for a given value of $d_1$, the steps are independent from one another, and can be performed in parallel.

The next few optimizations concern signatures.
First, we state the F5 criterion for matrix-weighted homogeneous systems.


\begin{theorem}[F5 criterion~\citep{Faugere-2002-F5}]
  \label{thm:F5-crit}
  Let $(m,i)$ be a signature appearing in the course of the algorithm.
  If $m$ lies in the ideal spanned by $f_1,\dots,f_{i-1}$, then the row inserted with signature $(m,i)$ is a linear combination of the rows inserted above.
\end{theorem}
\begin{proof}
  Assume that we chose $mf_{i}$ at line~\ref{algoline:new-row}.
  The hypothesis means that there exists polynomials $r, p_{1},\dots,p_{i-1}$ such that 
  \begin{equation}
    \label{eq:28}
    m + r = p_1f_1 + \dots + p_{i-1}f_{i-1},
  \end{equation}
  and all the terms in the support of $r$ are smaller than $m$.
  So
  \begin{equation}
    \label{eq:29}
    mf_{i} = f_{i}p_1f_1 + \dots + f_{i}p_{i-1}f_{i-1} - rf_{i}
  \end{equation}
  and the row with signature $(m,i)$ is linear combination of rows appearing above in the matrix, either with signature $(\mu, j)$ with $j < i$, or with signature $(\mu,i)$ with $\mu < m$.

  In particular, $V_{(m,i)} = V_{\prec (m,i)}$, and by the second loop invariant, it remains true regardless of the choice made at line~\ref{algoline:prev-elt}.
\end{proof}

This criterion can be used even with the algorithm simplified as described at the end of Section~\ref{sec:generic-matrix-f5}.
If we do not do this simplification and keep track of the signatures in the echelon form, we can also eliminate rows with the following Syzygy criterion, which essentially states that if the row associated with a signature is reduced to 0, the same will happen to all rows associated to a multiple of that signature.
\begin{theorem}[Syzygy criterion, e.g.~\citep{Roune-2012-PracticalGrobnerBasis}]
  \label{thm:syz-crit}
  Let $(m,i)$ be a signature appearing in the course of the algorithm at matrix-weighted degree $\mathbf{d}$, such that the row associated to the signature $(m,i)$ is a linear combination of the rows appearing above in the matrix at matrix-weighted degree $\mathbf{d}$.
  Let $\mu$ be a monomial.
  Then $(m\mu, i)$ is a linear combination of the rows appearing above in the matrix at matrix-weighted degree $\mathbf{d}' = \mathbf{d} + \pdeg_{\mathbf{W}}(\mu)$.
\end{theorem}
\begin{proof}
  By the second loop invariant of the algorithm, the rows above the row corresponding to $(m\mu,i)$ in $\Mac_{\mathbf{d}'}$ span $V_{\prec (m\mu,i)}$. 
  By the ordering on the signatures, this space contains the span of polynomials which can be written as $\mu p$, where $p \in V_{\prec (m,i)}$. 
  Since the polynomial associated with the row with signature $(m\mu,i)$ lies in the latter space, it lies in the former.
\end{proof}

The last optimization criterion of this section is new to the best of our knowledge.
It relies on the observation that at a given step, if all the signatures have a non-trivial gcd, then the matrix is the same as a matrix that has been considered at a previous step.
This is something that can appear in general weighted homogeneous and multihomogeneous cases too, but only in rare situations
On the other hand, for matrix-weighted homogeneous systems, it is inherent to the subdivision that some steps will have very few signatures, and very few ways to shift those signatures up.

\begin{proposition}
  \label{prop:sig-gcd}
  Let $(\mathbf{d},M_{\mathbf{d}}, S_{\mathbf{d}})$ be a step in the course of the algorithm.
  Let $\mu$ be the greatest common divisor of all the monomial parts of signatures in $S_{\mathbf{d}}$.
  If $\mu \neq 1$, then the step will not yield any new polynomial for the Gröbner basis.
\end{proposition}
\begin{proof}
  Let $\bm{\delta}=(\delta_{1},\dots,\delta_{k})$ be the matrix-weighted degree of $\mu$.
  By Remark~\ref{rmk:choice-prev-step}, we can assume that the rows of $\Mac_{\mathbf{d}}$ are formed by adding $mf_{i}$ for each signature $(m,i)$.
  Then all polynomials associated to those rows are of the form $\mu m' f_{i}$ where $m'f_{i}$ has matrix-weighted degree $\mathbf{d}-\bm{\delta}$.
  So after reducing $\Mac_{\mathbf{d}}$ to row echelon form, for any polynomial $g$ associated to a row in $\Mac_{\mathbf{d}}$, there exists $g'$ with matrix-weighted degree $\mathbf{d}-\bm{\delta}$ such that $g=\mu g'$.
  By the third loop invariant, $g'$ is reducible modulo $G$, and therefore so is $g$.
  By construction, $g$ will not be added to the basis.
\end{proof}


\subsection{Reducing to multihomogeneous systems}
\label{sec:converting-into-from}
\newcommand{\transmorph}{\mathcal{F}}
In the previous section, we presented an algorithm for computing a Gröbner basis for matrix-weighted homogeneous systems.
In this section, we will see that we can emulate the behavior of this algorithm by using state-of-the-art algorithms for multihomogeneous systems.

Those algorithms currently follow the structure of Algorithm~\ref{algo:matrixF5-phomo}, for the matrix of weights $\Wnb[*] = (W_0,W^{b}_{\mathbf{n},1},\dots,W^{b}_{\mathbf{n},k-1})$ with for $i>0$, $W^{b}_{\mathbf{n},i}$ as in Eq.~\eqref{eq:19}, and $W_{0} = (1,\dots,1) = \sum_{i=1}^{k} W_{\mathbf{n},i}^{b}$.
This matrix of weights is equivalent to the matrix of weights $\Wnb$.

\begin{definition}
  Let $\mathbf{W}$ be a non-negative $k$-system of weights in $\ZZ^{n \times k}$.
  Consider the algebra $B = K[Y_{1,1},\dots,Y_{n,1}, Y_{1,2}, \dots, ,\dots,Y_{n,k}]$ in $nk$ variables.
  Let $\mathbf{p}$ be the partition $kn = n + n + \dots + n$, and consider the corresponding  multigrading, defined by grouping the variables $Y_{\bullet,j}$ together for all $j \in \{1,\dots,k\}$.
  We define the morphism $\transmorph_{\mathbf{W}}: A \to B$ with
  \begin{equation}
    \label{eq:51}
    \transmorph_{\mathbf{W}}(f)(\mathbf{Y}) = 
    f\left( \mathbf{Y}_{1,\bullet}^{\mathbf{w}_{\bullet,1}},\dots , \mathbf{Y}_{n,\bullet}^{\mathbf{w}_{\bullet,n}} \right).
  \end{equation}
\end{definition}
 
\begin{proposition}
  \label{prop:basic-phi}
  The map $\transmorph_{\mathbf{W}}$ is an injective morphism and sends matrix-weighted homogeneous polynomials in $A$ with matrix-weighted degree $\mathbf{d}$ onto multihomogeneous polynomials in $B$ with multidegree $\mathbf{d}$.
  Furthermore, this morphism sends the order $<_{\WMgrevlex}$ on the order $<_{\WMgrevlex[\Wpb]}$.
\end{proposition}

\begin{example}
  With the matrix of weights $\mathbf{W}=
  \begin{psmallmatrix}
    1 & 1 & 1\\
    1 & 2 & 3
  \end{psmallmatrix}$,
  \begin{equation}
    \label{eq:52}
    \transmorph_{\mathbf{W}}(X_{1}X_{2}^{2} + X_{1}^{2}X_{3}) = Y_{1,1}Y_{1,2}Y_{2,1}^{2}Y_{2,2}^{4} + Y_{1,1}^{2}Y_{1,2}^{2}Y_{3,1}Y_{3,2}^{3},
  \end{equation}
  and it is multihomogeneous with degree $3$ in $Y_{1,1},Y_{2,1},Y_{3,1}$, and degree $5$ in $Y_{1,2},Y_{2,2},Y_{3,2}$.
\end{example}

The main observation for the computation of Gröbner bases is that this morphism commutes with S-polynomials.
\begin{proposition}
  \label{prop:spol-stable-phi}
  Let $f$ and $g$ be two polynomials.
  Then
  \begin{equation}
    \label{eq:23}
    \SPol(\transmorph_{\mathbf{W}}(f), \transmorph_{\mathbf{W}}(g)) = \transmorph_{\mathbf{W}}(\SPol(f,g)).
  \end{equation}
\end{proposition}

This property makes it possible to compute Gröbner bases through $\transmorph_{\mathbf{W}}$.
\begin{theorem}
  Let $F=(f_{1},\dots,f_{m})$ be a tuple of polynomials in $A$.
  Endow the algebra $A$ with the monomial order $<_{\WMgrevlex}$.
  Let $\transmorph_{\mathbf{W}}(F)=  (\transmorph_{\mathbf{W}}(f_1),\dots,\transmorph_{\mathbf{W}}(f_{m}))$ in $B$.
  Endow the algebra $B$ with the matrix of weights $\mathbf{W}_{\mathbf{p}}^{\rm b}$ (\emph{cf.} Eq.~\eqref{eq:19}).
  Then the following operations commute:
  
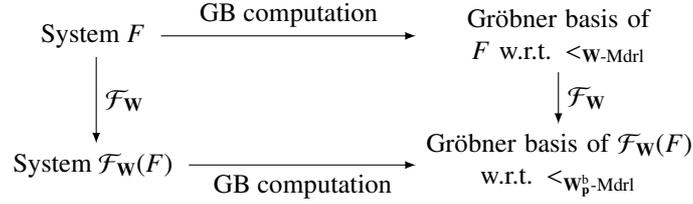
\begin{center}
    \begin{tikzpicture}[every node/.style={font=\rm}]
      \node[matrix of nodes, ampersand replacement=\&,
      nodes={align=center, anchor=center},
      column 2/.style={nodes={text width=3.5cm}}
      ] (diagram)
      {System $F$ \&[3cm]
        Gröbner basis of $F$ w.r.t. $<_{\WMgrevlex}$ \\[0.7cm]
      System $\transmorph_{\mathbf{W}}(F)$ \&
      {Gröbner basis of $\transmorph_{\mathbf{W}}(F)$ \\
        w.r.t. $<_{\WMgrevlex[\Wpb]}$
      }  \\} ;
    \def\myshift{3pt}
    \begin{scope}[every path/.style={draw,-latex}]
      \path (diagram-1-1) -- (diagram-1-2) node[pos=0.5,above] {GB computation};
      \path (diagram-1-1) -- (diagram-2-1) node[pos=0.5,right] {$\transmorph_{\mathbf{W}}$};
      \path (diagram-2-1) -- (diagram-2-2) node[pos=0.5,below] {GB computation};
      \path (diagram-1-2) -- (diagram-2-2) node[pos=0.5,right] {$\transmorph_{\mathbf{W}}$};
    \end{scope}
    \end{tikzpicture}
  \end{center}
\end{theorem}

\begin{proof}
  It is a consequence of the fact that a Gröbner basis can be computed by a sequence of S-polynomials and reductions (themselves a particular case of S-polynomial).
  By Proposition~\ref{prop:spol-stable-phi}, the result of a sequence of such calculations, run on polynomials of the form $\transmorph_{\mathbf{W}}(f)$, will be polynomials in the image of $\transmorph_{\mathbf{W}}$.
  And by Proposition~\ref{prop:basic-phi}, the orders are compatible, so that applying $\transmorph_{\mathbf{W}}$ preserves the propriety of being a Gröbner basis.
\end{proof}

So we can compute a Gröbner basis of $\transmorph_{\mathbf{W}}(F)$ using state-of-the-art algorithms for multihomogeneous Gröbner bases, and invert $\transmorph_{\mathbf{W}}$ to recover the wanted basis.

In practice, as mentioned, algorithms for multihomogeneous systems compute a basis for the $\Wpb[*]$ ordering, which is equivalent to the $\Wpb$ ordering.
By Proposition~\ref{prop:equiv-w-gb}, it does not change the result of a Gröbner basis computation if the input is multihomogeneous, and it amounts to a mere reordering of the matrices considered in the course of the algorithm.

The algorithm does not take into account the inherent sparsity of systems arising from the application of $\transmorph_{\mathbf{W}}$.
This leads to both larger matrices and to more matrices.
This can be handled by restricting the computation to monomials appearing in the image of $\transmorph_{\mathbf{W}}$.
In doing so, there are only finitely monomials of given $W_{1}$-degree, which allows the algorithm to run directly with the weights $\Wpb$.
With those changes, the multihomogeneous algorithms encodes the same operations, in the same order, as the matrix-weighted homogeneous algorithm.

This is similar to the weighted homogeneous case, where taking advantage of the structure in matrix algorithms requires restricting to monomials in the image of the transformation.

\subsection{Sparse Gröbner basis algorithms}
\label{sec:sparse-grobner-basis}

One last approach worth mentioning is the use of sparse Gröbner basis algorithms and in particular their specialization to multihomogeneous systems.
We recall basic facts about the sparse variant of the Matrix-F5 algorithm in that case.
Let $\mathbf{n}$ be a partition of $n$, and let $F$ be a multihomogeneous system of polynomials with multidegree $\mathbf{d}$.
For simplicity, we assume that all polynomials have the same multidegree $\mathbf{d}$.
The sparse Matrix-F5 algorithm can be emulated by running Algorithm~\ref{algo:matrixF5-phomo} with \textsf{AlgoSteps} a function returning, for all $d \in \{1,\dots,\dmax\}$, the monomials obtained by products of $d$ monomials of multidegree $\mathbf{d}$, and for signatures the pairs $(m,i)$ with $m$ a product of $d-1$ monomials of multidegree $\mathbf{d}$.
We refer the reader to the seminal papers~\cite{Faugere-2014-sparse, Bender-2019-grobner-basis-over} for more details.

This algorithm does not in general compute a Gröbner basis, but a \emph{sparse Gröbner basis}.
The notion of sparse Gröbner basis is relative to the choice of the monomials in the computation.
It is in particular suitable for listing the solutions of a system that is zero-dimensional in the multihomogeneous sense.
However, systems arising from the $\transmorph_{\mathbf{W}}$ transformation will almost never be zero-dimensional: a system of $n$ equations in $n$ variables becomes a system of $n$ equations in $kn$ variables.

It is also possible to run the sparse Matrix-F5 algorithm directly on the matrix-weighted homogeneous system.
In practice, it means that given a system $F$ of $\mathbf{W}$-matrix-weighted homogeneous polynomials with the same $\mathbf{W}$-matrix-weighted degree $\mathbf{d}$, the function \textsf{AlgoSteps} function would return, at each degree $d$, the monomials obtained by the products of $d$ monomials of matrix-weighted degree $\mathbf{d}$, and similarly for the signature.
If the system is zero-dimensional in the sparse sense, this allows one to compute a sparse Gröbner basis, and then to recover the solutions using the sparse FGLM algorithm.

\subsection{Affine systems}
\label{sec:affine-systems}

An algorithm following the $\mathbf{W}$-degree for systems that are not matrix-weighted homogeneous can be useful from at least two perspectives.
First, similar to the usual cases, the polynomials in a system may have a large matrix-weighted homogeneous component of maximal matrix-weighted degree (say, for a lexicographical comparison),
so that using the structure to guide the computation is advantageous.
Second, with the observation at the end of Section~\ref{sec:matrix-gr-reverse-lexic}, it is possible to represent \emph{any} monomial order as a $\WMgrevlex$ order for some $\mathbf{W}$.
In this case, for each $\mathbf{W}$-degree, there is exactly one monomial, and so all polynomials are affine.

Assuming that $W_{1}$ is positive, Algorithm~\ref{algo:matrixF5-phomo} can be adapted to compute a Gröbner basis for those orders, by also adding lower matrix-weighted degree monomials to $M$.

\begin{remark}
  If $W_{1}$ is not positive, by Remark~\ref{rmk:equiv-weights-orders}, it is not enough to find an equivalent system of weights such that $W_{1}$ is positive.
\end{remark}

However, adapting the criteria to deal with those affine situations is complicated.
Even in the homogeneous case (standard grading), the Matrix-F5 algorithm does not behave nicely with affine polynomials.
Indeed, it is proved~\citep{Eder-2013-AnalysisInhomogeneousSignaturebased} that whenever the degree of the polynomials falls, the resulting polynomials can no longer be used for the F5 criterion and that conversely, the F5 criterion may not catch all reductions to zero even for regular sequences.

For this reason, affine systems are usually studied under the stronger hypothesis of \emph{regularity in the affine sense}~\citep{Bardet-2004-thesis}, which encodes that the highest degree components of the polynomials form a regular sequence and thus that the F5 criterion can avoid all degree falls.

This approach could work for polynomials with a large highest $\mathbf{W}$-degree component and whose highest matrix-weighted degree components satisfy the regularity conditions in Section~\ref{sec:hilb-mult-regul}.
However, when looking at the matrices of weights for a fixed weight-matrix order, rather than letting the structure of the system guide the choice of the weights, it is expected that degree falls will happen.
Besides the obvious case of $n$ linearly independent systems of weights, for example, for the weights defining an elimination ordering, any actual elimination will be a degree fall.

In the usual case with the standard degree, when degree falls cannot be avoided, there are two approaches.
First, one can homogenize the system (see~\cite[Sec.~4.3.A]{Kreuzer-2005-ComputationalCommutativeAlgebra}).
In that case, degree falls cannot happen, and the F5-criterion can avoid reductions to zero.
This approach would work in our case, using $k$ different homogenization variables.
However, just like in the classical case, this leads to working with polynomials with a $\mathbf{W}$-degree much higher than necessary.

Note that with the proper choice of a monomial order (see~\cite[Sec.~4.4]{Kreuzer-2005-ComputationalCommutativeAlgebra}), homogenizing the system is equivalent to considering all monomials and all signatures \emph{up to} each $\mathbf{W}$-degree, coordinate-wise.
The advantage of the latter perspective or, equivalently, of greedily dehomogenizing the computed polynomials, is that it makes it possible to follow the degree falls and insert them in the matrices computed at the appropriate degree.
This invalidates the criteria, but it avoids computing polynomials with a degree higher than necessary and it is the most efficient approach in practice for standard graduations.

\section{Regularity and Hilbert multiseries}
\label{sec:hilb-mult-regul}

\subsection{Definitions}
\label{sec:definitions-1}

Analyzing the complexity of the Matrix-F5 algorithm requires bounding the minimal degree bound at which a truncated Gröbner basis is complete, which gives a bound on the number of degree steps taken by the algorithm, as well as the size of the matrices built at each step.
A valuable tool for that purpose is the Hilbert series of ideals.
Indeed, for a lot of structures, one can identify suitable definitions of the regularity of sequences, for which it is possible to compute the Hilbert series using elementary series operations and to read complexity parameters off of it.
Those regularity properties are intrinsically connected to the algorithmic process, as they correspond to rank properties of the computed matrices, and to a characterization of the reductions to zero.

Among such properties, one can list regular sequences, semi-regular sequences, and sequences in Noether position.
Those properties have been notably used in the homogeneous, weighted homogeneous, and sparse cases.

In the multihomogeneous case, the definition of regularity is significantly more complicated than for $\NN$-graduations, as it requires characterizing divisors of zero which are \emph{always} present, regardless of the ideal.
This is done only in a few specific cases, including the bilinear case.

As can be expected, the situation in the matrix-weighted case is not easier.
In this section, we give different possible definitions of regularity, and we discuss their relations with the known cases.

\begin{definition}
  Let $\WW \in \NN^{k \times n}$ be a non-negative matrix of weights, and let $A = K[X_{1},\dots,X_{n}]$.
  Let $F = (f_{1},\dots,f_{r}) \in A^{r}$ be a sequence of \multiwhomo polynomials, with respective \multiwdeg $\mathbf{d}_{i}$.
  For all $i \in \{1,\dots,r\}$, let $I_{i}$ be the ideal $\langle f_{1},\dots,f_{i}\rangle$.

  For all $i \in \{2,\dots,r\}$, and for all \multiwdeg $\mathbf{d}$, we consider the multiplication map
  \begin{equation}
    \label{eq:3}
    \begin{array}{rrcl}
      \phi_{i}^{(\mathbf{d})}: & (A/I_{i-1})_{\mathbf{d}} & \to & (A/I_{i-1})_{\mathbf{d}+\mathbf{d}_{i}} \\
      & f & \mapsto & f_{i}f.
    \end{array}
  \end{equation}
  Let $f \in A_{\mathbf{d}}$ be such that $f_{i}f \in I_{i-1}$, we say that:
  \begin{enumerate}
    \item $f$ is a \emph{trivial divisor of $0$ in $A/I_{i-1}$} if $f \in I_{i-1}$;
    \item $f$ is a \emph{semi-trivial divisor of $0$ in $A/I_{i-1}$} if $\Mdeg_{\mathbf{W}}(f) = \mathbf{d}$ and $\phi_{i}^{(\mathbf{d})}$ is surjective\footnote{This definition only loosely depends on $f$, via its $\mathbf{W}$-degree $\mathbf{d}$. Indeed, if $\phi_{i}^{(\mathbf{d})}$ is surjective, then any divisor of $0$ at $\mathbf{W}$-degree $\mathbf{d}$ is a semi-trivial divisor of $0$. The naming is intentional, as the focus is on the individual divisors of zero.};
    \item $f$ is an \emph{eliminable divisor of $0$ in $A/I_{i-1}$} if the leading term of $f$ is divisible by the leading term of a semi-trivial divisor of $0$ in $A/I_{i-1}$.
  \end{enumerate}
  We also say that:
  \begin{enumerate}
    \item $F$ is a \emph{regular sequence} if for all $i$, all divisors of $0$ in $A/I_{i-1}$ are trivial;
    \item $F$ is a \emph{semi-regular sequence} if for all $i$, all divisors of $0$ in $A/I_{i-1}$ are trivial or semi-trivial (i.e., for all $\mathbf{d}$, $\phi_{i}^{\mathbf{d}}$ is either injective or surjective);
    \item $F$ is a \emph{weakly regular sequence} if for all $i$, all divisors of $0$ in $A/I_{i-1}$ are trivial, semi-trivial or eliminable.
\end{enumerate}
\end{definition}

In Algorithm~\ref{algo:matrixF5-phomo}, different types of divisors of zero correspond to different types of reductions to zero:
\begin{enumerate}
  \item if $f$ is a trivial divisor of $0$ in $A/I_{i-1}$, then $(\lm(f),i)$ is a signature eliminated by the F5 criterion;
  \item if $f$ is a semi-trivial divisor of $0$ in $A/I_{i-1}$, then the matrix $\Mac_{d}$, after adding the signatures in $i$, has more rows than columns, and has full rank; if there are as many pivots as columns above the row associated to signature $(\lm(f),i)$, it is possible to avoid that reduction to zero;
  \item if $f$ is an eliminable divisor of $0$ in $A/I_{i-1}$, then $(\lm(f),i)$ is a signature eliminated by the syzygy criterion.
\end{enumerate}

We make a few notes on the relations between this definition and pre-existing notions:
\begin{enumerate}
  \item The definition of regular sequences does not depend on the grading and is the same as the usual definition, including in the homogeneous~\citep{Lazard-1983} and weighted homogeneous \citep{Faugere-2013-whomo-1} cases.
  \item The definition of semi-regular sequences does depend on the grading. For $\NN$-gradings, it specializes to the usual definition in the homogeneous case~\citep{Bardet-2004, Pardue-2010} and some weighted homogeneous cases~\citep{Faugere-2016-whomo-2}.
  In those cases, if $r \leq n$, there cannot exist semi-trivial divisors of $0$, and weak regularity, semi-regularity, and regularity are equivalent.
  \item If $r \geq i > n$, in the homogeneous and some weighted homogeneous cases, for semi-regular sequences, it is known that there exists a (weighted) degree $\delta^{(i)}$ such that for all $d < \delta^{(i)}$, $\phi_{i}^{(d)}$ is injective, and for all $d \geq \delta_{i}$, $\phi_{i}^{(d)}$ is surjective.
  So in those cases, weak regularity implies semi-regularity.
\end{enumerate}

For known $\NN$-graduations, the situation is therefore simple: if $r \leq n$, all three properties are equivalent, and if $r > n$, regular sequences do not exist but the other two properties are equivalent.

For $\NN^{k}$-graduations, things are more complicated, as the following examples illustrate.

\begin{example}[A sequence which is semi-regular and not regular, with $r \leq n$]
  \label{ex:semi-reg-not-reg}
  Consider the algebra $K[x,y,z]$ graded by the matrix of weights $\mathbf{W} =
  \begin{psmallmatrix}
    1 & 1 & 1 \\
    1 & 2 & 3
  \end{psmallmatrix}
  $ and the sequence $(x^{2}y^{2} + x^{3}z, x^{2}y^{2} - x^{3}z, x^{2}y^{2} + 2x^{3}z)$.
  All polynomials are $\mathbf{W}$-homogeneous with $\mathbf{W}$-degree $(4,6)$.
  
  Clearly, the $3$ polynomials are not linearly independent, so the sequence is not regular: $1$ is a divisor of $0$ in $A/I_{2}$.
  But since there are only $2$ monomials of $\mathbf{W}$-degree $(4,6)$ in $A$, the quotient $A/I_{2}$ is $0$, and therefore $\phi_{3}^{(0,0)}$ is surjective. So $1$ is a semi-trivial divisor of $0$.

  This rank defect persists in later matrices, so all the divisors of $0$ for that sequence are semi-trivial.
\end{example}

\begin{example}[A sequence which is weakly regular but not semi-regular, with $r \leq n$]
  Consider the algebra $K[x,y,z]$ graded by the matrix of weights $\mathbf{W} =
  \begin{psmallmatrix}
    1 & 1 & 5\\
    1 & 2 & 5
  \end{psmallmatrix}
  $ and the sequence $(x^{10}+z^{2}, xy, xy)$.
  All polynomials are $\mathbf{W}$-homogeneous with respective $\mathbf{W}$-degree  $(10,10)$, $(2,3)$ and $(2,3)$.

  Obviously, for all $\mathbf{W}$-degree $\mathbf{d}$, the map $\phi_{3}^{\mathbf{d}}$ is the 0 map, which means that for all $f \in A$, $f$ is a divisor of $0$ in $A/I_{2}$.

  At $\mathbf{W}$-degree $(2,3)$, there is only one monomial in $A$, so the map $\phi_{3}^{(0,0)}$ is surjective, and $1$ is a semi-trivial divisor of $0$ in $A/I_{2}$.

  At $\mathbf{W}$-degree $(6,7)$, there are two monomials: $x^{5}$ and $z$, and the latter clearly does not lie in $I_{2}=I_{3}$.
  So the map $\phi_{3}^{(4,4)}$ is not surjective, and therefore the corresponding cofactor $x^{4}$ is not a semi-trivial divisor of $0$ in $A/I_{2}$.
  It is however eliminable since it is divisible by $1$ which is a semi-trivial divisor of $0$.
  
  Since the first two polynomials form a regular sequence, and all divisors of $0$ in $A/I_{2}$ will be divisible by $1$, the sequence is weakly regular.
\end{example}

\begin{example}
  The previous example is intentionally simple, to the point that one could argue that such a sequence should not be considered ``regular'' in any way.
  Here is how to build a similar but non-trivial example, in the weighted homogeneous case: consider a system of weights $W=(1,w_{2},\dots,w_{n})$, and a semi-regular, non regular, sequence $F=(f_{1},\dots,f_{r})$ in $K[x_{1},\dots,x_{n}]$, with respective $W$-degree $d_{1},\dots,d_{n}$.
  Let $j,\delta$ be such that there exists a semi-trivial divisor $f$ of $0$ in $A/\langle f_{1},\dots,f_{j-1}\rangle$ with $\deg_{W}(f) = \delta-d_{j}$.
  For simplicity, further assume that $f$ and $xf$ do not lie in $I_{j-1}$.
  
  Now consider the algebra $A'=K[x_{1},\dots,x_{n},y]$ graded by the system of weights $W'=(1,w_{2}\dots,w_{n},\delta+1)$.
  In this new algebra, let $f_{0}=x_{1}^{2\delta + 2} + y^{2}$ and consider the extended sequence $F'=(f_0, f_{1}, \dots, f_{r})$.
  Let $I'_{i} = \langle f_{0},\dots,f_{i} \rangle \subseteq A'$ for $i \in \{0,\dots,r\}$, with additionally $I'_{-1}=0$.
  Since $y$ has $W$-degree $\delta+1$, for all $d \leq \delta$ and for all $i \in \{-1,\dots,r\}$, there are canonical isomorphisms between $A_{d}$ and $A'_{d}$, between $(I_i)_{d}$ and $(I'_{i})_{d}$, and between $(A/I_{i})_{d}$ and $(A'/I'_{i})_{d}$.
  In particular, since $f$ has $W$-degree $\delta - d_{j}$, $f$ is still a semi-trivial divisor of zero in $A'/I'_{j-1}$.

  Therefore, $xf f_{j} \in (I_{j-1})_{\delta+1}$, and $xf$ is a divisor of zero in $A/I_{j-1}$.
  By assumption, $xf \notin I_{j-1}$. Since $I_{j-1} = I'_{j-1} \cap A$, also $xf \notin I'_{j-1}$, and so this divisor of zero is not trivial.
  But at $W$-degree $\delta+1$, the algebra $A'$ contains the monomial $y$ which cannot lie in the vector space $I'_{j}$, so $\phi_{j}^{(\delta-d_{j}+1)}$ cannot be surjective.
  Therefore, the divisor of zero $xf$ is not semi-trivial, but it is eliminable.
  
  Similarly as in the previous example, since the polynomial $f_{0}$ is not a divisor of zero modulo the ideal spanned by $f_{1},\dots,f_{r}$, any element $f$ such that $f f_{i} \in I'_{i_1}$ with a non-trivial cofactor in $f_{0}$ will be a trivial divisor of $0$.
  Therefore all divisors of $0$ are trivial, semi-trivial, or eliminable, and the sequence is weakly regular.
\end{example}

\subsection{Hilbert multiseries}
\label{sec:hilbert-multiseries}

We recall the definition of the Hilbert series in $\ZZ^{k}$-graded settings.

\begin{definition}[{\citet[Def.~5.8.11]{Kreuzer-2005-ComputationalCommutativeAlgebra}}]
  Let $k,n \in \NN$, let $\mathbf{W} \in \ZZ^{k \times n}$ be a size-bounded $k$-matrix of weights.
  We consider the algebra $A = K[\mathbf{X}]$ with the $\mathbf{W}$-matrix-grading.
  Let $I \subset A$ be a $\mathbf{W}$-matrix-weighted homogeneous ideal.
  The \emph{Hilbert multiseries} of $A/I$ is the formal series
  \begin{equation}
    \label{eq:1}
    \HS_{A/I}(T_{1},\dots,T_{k}) = \sum_{\mathbf{d}=(d_{1},\dots,d_{k}) \in \NN^{k}} \dim_{K}((A/I)_{\mathbf{d}}) T_{1}^{d_{1}}
    \cdots T_{k}^{d_{k}}.
  \end{equation}
\end{definition}

\begin{proposition}[{\citet[Th.~5.8.15]{Kreuzer-2005-ComputationalCommutativeAlgebra}}]
  The Hilbert multiseries of the algebra $A$ is 
  \begin{equation}
    \label{eq:2}
    \HS_{A}(\mathbf{T}) = \frac{1}{\left(1-T_{1}^{w_{1,1}}\cdots T_{k}^{w_{k,1}}\right)\cdots \left(1-T_{1}^{w_{1,n}}\cdots T_{k}^{w_{k,n}}\right)}
     = \frac{1}{(1-\mathbf{T}^{W_{\bullet, 1}})\cdots (1-\mathbf{T}^{W_{\bullet, n}})}.
  \end{equation}
\end{proposition}
\begin{proof}
  The dimension $\dim(A_{\mathbf{d}})$ is equal to the number of monomials of $\mathbf{W}$-degree $\mathbf{d}$.
  The sum of all monomials in $A$ can be expressed as a fraction
  \begin{equation}
    \label{eq:12}
    \sum_{\bm{\alpha} \in \NN^{n}} \mathbf{X}^{\bm{\alpha}} = \frac{1}{(1-X_{1})\cdots (1-X_{n})},
  \end{equation}
  and the result follows after substituting $X_{i}$ by $\mathbf{T}^{\mathbf{w_{\bullet,i}}}X_{i}$, grouping by degree in $\mathbf{T}$ and evaluating at $X_{i}=1$.
\end{proof}

The next theorem describes a way to compute the Hilbert multiseries of a quotient by an ideal generated by a regular or semi-regular sequence.

\begin{theorem}
  With the same notations, let $(f_{1},\dots,f_{r}) \in A^{r}$ be a sequence of $\textbf{W}$-matrix-weighted homogeneous polynomials, spanning an ideal $I$, with respective \multiwdeg $\mathbf{d}_{i}$.
  For $i \in \{0,\dots,m\}$, let $I_{i}$ be the ideal spanned by $(f_{1},\dots,f_{i})$, with $I_{0}=A$ and $I_{r}=I$.

  Then the sequence $F$ is regular iff for all $i$,
  \begin{equation}
    \label{eq:45}
    \HS_{A/I_{i}}(\mathbf{T})
  = (1-\mathbf{T}^{\mathbf{d}_{i}})\HS_{A/I_{i-1}}(\mathbf{T}) 
= \frac{(1-\mathbf{T}^{\mathbf{d}_{1}})\cdots (1-\mathbf{T}^{\mathbf{d}_{m}})}{(1-\mathbf{T}^{W_{\bullet, 1}})\cdots (1-\mathbf{T}^{W_{\bullet, n}})}.
  \end{equation}
  The sequence $F$ is semi-regular iff for all $i$,
  \begin{equation}
    \label{eq:45b}
    \HS_{A/I_{i}}(\mathbf{T}) = \left\lfloor (1-\mathbf{T}^{\mathbf{d}_{i}})\HS_{A/I_{i-1}}(\mathbf{T}) \right\rfloor,
  \end{equation}
  where $\lfloor  S \rfloor$ is obtained by removing negative coefficients from $S$.
\end{theorem}
\begin{proof}
  The proof is similar to the usual case, so we only outline the argument.
  Let $c_{\mathbf{d},i}=\dim(A/I_{i})_{\mathbf{d}}$.
  The multiplication map $\phi_{i}^{(\mathbf{d})}$ defines an exact sequence
  \begin{equation}
    \label{eq:46}
    0 \to K_{\mathbf{d}} \to (A/I_{i-1})_{\mathbf{d}} \to (A/I_{i-1})_{\mathbf{d}+\mathbf{d}_{i}}
    \to (A/I_{i})_{\mathbf{d}+\mathbf{d}_{i}} \to 0
  \end{equation}
  where $K_{\mathbf{d}} = \Ker(\phi_{i}^{(\mathbf{d})})$.
  This gives that
  \begin{equation}
    \label{eq:47}
    c_{\mathbf{d}+\mathbf{d}_{i},i} = c_{\mathbf{d}+\mathbf{d}_{i}, i-1} - c_{\mathbf{d}, i-1} - \dim(K_{\mathbf{d}}).
  \end{equation}
  $F$ is a regular sequence iff $K_{\mathbf{d}}=0$.
  $F$ is a semi-regular sequence if and only if either $K_{\mathbf{d}}=0$ (if $\phi_{i}^{(\mathbf{d})}$ is injective) or $\dim_{K}((A/I_{i})_{\mathbf{d}+\mathbf{d}_{i}})=0$ (if $\phi_{i}^{(\mathbf{d})}$ is surjective, i.e., $c_{\mathbf{d},i-1}-c_{\mathbf{d}+\mathbf{d}_{i},i-1}\leq 0$
  ).
  In either case, multiplying Eq.~\eqref{eq:47} by $\mathbf{T}^{\mathbf{d}-\mathbf{d}_{i}}$ and summing over $\mathbf{d}$ gives the result.
  The closed form for the regular case follows, by induction.
\end{proof}

\begin{remark}
  For semi-regular sequences, contrary to the usual case, it is not necessarily true that one can remove all coefficients after the \emph{first} negative coefficient, or even that one can remove coefficients of monomials which are multiples of that monomial with a negative coefficient, as the following example shows.
\end{remark}

\begin{example}
  \label{ex:semi-reg-hs}
  Consider again the sequence from Example~\ref{ex:semi-reg-not-reg}, with 3 polynomials of the form $x^{2}y^{2} + \alpha x^{3}z$, $\mathbf{W}$-homogeneous with $\mathbf{W}$-degree $(4,6)$ for $\mathbf{W} = \begin{psmallmatrix}
    1 & 1 & 1 \\
    1 & 2 & 3
  \end{psmallmatrix}$.
  The first negative coefficient arises when computing $\HS_{A/I_{3}}(T,U)$, at multidegree $(4,6)$, corresponding to the semi-trivial divisor of $0$ at that multidegree.
  Accordingly, the quotient $A/I_{3}$ at that multidegree is trivial, all the monomials are in the ideal.
  However, at multidegree $(5,10)$, the algebra contains the monomials $y^{5}$, $xy^{3}z$ and $x^{2}yz^{2}$, which do not lie in the ideal.
  Note also that while subsequent powers of $xy^{3}z$ and $x^{2}yz^{2}$ lie in the ideal, none of the powers of $y^{5}$ does, so the Hilbert series will have a positive coefficient at all bidegrees $(5k,10k)$.

  The problem here is that even if the ideal contains all monomials at a given $\mathbf{W}$-degree, there may be monomials of higher $\mathbf{W}$-degree which are not divisible by any monomial at that degree.
  This can be compared to the case of a single system of weights studied in~\cite{Faugere-2016-whomo-2}.
  In that case, the same phenomenon can appear if the weights are chosen without any assumption, see e.g. Figure 2 in that work.

  Let us examine more precisely the Hilbert multiseries in the present example.
  Up to the second polynomial, the sequence is regular, so
  \begin{equation}
    \label{eq:7}
    \HS_{A/I_{2}}(T,U) = \frac{(1-T^{4}U^{6})^{2}}{(1-TU)(1-TU^{2})(1-TU^{3})}.
  \end{equation}
  We plotted in Figure~\ref{fig:coeff_hs} the coefficients of
  \begin{equation}
    \label{eq:10}
    (1-T^{4}U^{6})\HS_{A/I_{2}} = \frac{(1-T^{4}U^{6})^{2}}{(1-TU)(1-TU^{2})(1-TU^{3})},
  \end{equation}
  highlighting the negative coefficients in grey.
  For instance, the bottom-left-most grey coefficient corresponds to the aforementioned multidegree $(4,6)$.
  The Hilbert multiseries $\HS_{A/I_{3}}$ is given by the sum of all the monomials with a positive coefficient.
  
  Those dropped coefficients clearly do not represent a truncation of the series.
  In this simple example, the removed coefficients appears to be a shift of the cone of all monomials in the algebra. 
  We leave the question of whether this phenomenon can be characterized and generalized, to future work.
  \end{example}

\begin{figure}
  \centering

  \begin{tikzpicture}[x=0.6cm,y=0.6cm]
    \draw[thick,->] (0,0) -- (10.5,0) node[right] {$\deg_{T}$};
    \draw[thick,->] (0,0) -- (0,10.5) node[right] {$\deg_{U}$};
    \begin{scope}[every node/.style={fill=white}]
      \foreach \x/\y/\c in {
        0/0/1, 1/1/1, 1/2/1, 2/2/1, 1/3/1, 2/3/1, 3/3/1, 2/4/2, 3/4/1, 2/5/1, 4/4/1, 3/5/2, 2/6/1, 4/5/1, 3/6/2, 5/5/1, 4/6/\textcolor{gray}{-1}, 3/7/2, 5/6/1, 4/7/2, 3/8/1, 6/6/1, 5/7/\textcolor{gray}{-1}, 4/8/3, 3/9/1, 6/7/1, 5/8/\textcolor{gray}{-1}, 4/9/2, 7/7/1, 6/8/\textcolor{gray}{-1}, 4/10/2, 7/8/1, 6/9/\textcolor{gray}{-1}, 5/10/3, 8/8/1, 7/9/\textcolor{gray}{-1}, 6/10/\textcolor{gray}{-2}, 8/9/1, 7/10/\textcolor{gray}{-1}, 9/9/1, 8/10/\textcolor{gray}{-1}, 9/10/1, 10/10/1
}
      {\node at (\x,\y) {\c};}
    \end{scope}
  \end{tikzpicture}
  
  \caption{Coefficients of the series $(1-T^{4}U^{6})\HS_{A/I_{2}} $ (in black and gray) and $\HS_{A/I_{3}}(T,U)$ (in black only), in Example~\ref{ex:semi-reg-hs}}
  \label{fig:coeff_hs}
\end{figure}



Contrary to the usual case, it is also not clear how to fold back the recurrence \eqref{eq:45} into a closed form formula for the Hilbert multiseries of a semi-regular sequence.
Even in the case of a single system of weights, fully characterizing the Hilbert series of semi-regular sequences is an open question, which is only answered under some assumptions on the weights and the weighted degrees~\cite[Th. 15]{Faugere-2016-whomo-2}.
Getting a closed form for the Hilbert series of semi-regular and weakly regular sequences, as well as extracting from those series complexity parameters such as the number of solutions or the degree of the polynomials in the Gröbner basis, should be the object of future research.

\subsection{Genericity}
\label{sec:gener-regul-sequ}


We now prove a partial statement concerning the genericity of the above properties.
Recall that a property $\mathcal{P}$ of polynomial systems is said to be \emph{generic} among a subset $S$ of $A^{r}$ if there exists a non-empty Zariski-open subset $U \subseteq S$ such that all elements of $U$ have the property $\mathcal{P}$.


\begin{theorem}
  Let $\mathbf{d}_{1},\dots,\mathbf{d}_{r}$ be matrix-weighted degrees such that there exists at least one regular (resp. semi-regular, resp. weakly regular) matrix-weighted homogeneous sequence with respective matrix-weighted degree $\mathbf{d}_{i}$.
  Then, among matrix-weighted homogeneous sequences with respective matrix-weighted degree $\mathbf{d}_{i}$, regular (resp. semi-regular, resp. weakly regular) sequences are generic.
\end{theorem}
\begin{proof}
  In all three cases, the property is a statement on the Macaulay matrices built by Algorithm~\ref{algo:matrixF5-phomo}, at each matrix-weighted degree.
  Each of the coefficients, both before and after reducing to row-echelon form, can be written as a rational function in the coefficients of the input polynomial, in the Zariski-open set defined by the denominators not vanishing.
  So the Macaulay matrix has a generic shape.

  All three properties state that the Macaulay matrices, after removing rows with the criteria, have full rank.
  Since the rank is upper semi-continuous, if there exists a system for which the Macaulay matrix at matrix-weighted degree $\mathbf{d}$ has maximal rank, so does the generic Macaulay matrix.
  This implies that the properties are generic.
\end{proof}

In general, the properties are not always generic, as the third example in Table~\ref{tab:exp1} shows.
%
  %
It is not unexpected, as it was already the case with weighted homogeneous systems that regular sequences are not generic when there do not exist enough monomials at the wanted weighted degrees.
Based on experiments, we expect however that it is possible to identify large classes of matrix-weighted degrees such that some of the regularity properties are generic, in particular for non-degenerate cases with $r \leq n-k$.

\section{Experimental results}
\label{sec:applications}

We wrote a toy implementation\footnote{\url{https://gitlab.com/thibaut.verron/matrix-homo-gb}} of Algorithm~\ref{algo:matrixF5-phomo} specialized to the matrix-weighted homogeneous case.
In Table~\ref{tab:exp1}, we briefly report on the largest total degree reached, the number and size of the matrices, the number of parallelized steps, and the number of reductions to zero, for generic systems.
The experiments were run with a limit of 40 parallel processes.
We report the results for three algorithms: Algorithm~\ref{algo:matrixF5-phomo} with and without filtering out steps with a non-trivial GCD, and the state of the art (following the $W_{1}$-homogeneous structure).
If the values in a row are indicated as ``$\geq$'', it means that the computation was run with the given value of $\dmax$, but that the result is only a truncated Gröbner basis.
The number of reductions to zero is the total number of reductions done by the algorithm with the two criteria, and the number in parentheses is the number of rank defects in the Macaulay matrices (syzygies that are neither trivial, semi-trivial or eliminable).

Beyond the usual signature criteria, it appears that eliminating steps with a non-trivial GCD eliminates quite many small matrices and results in a significant speed-up. This may be explained by eliminating the overhead of spawning parallel processes for trivial computations.

\newcommand{\truncated}[1]{$\geq #1$}

\begin{table*}[t]
  \centering
  \caption{Experimental data on generic matrix-weighted homogeneous systems. \\
    Each row is a system of $r$ polynomials, $\mathbf{W}$-homogeneous with the same $\mathbf{W}$-degree. We run Algorithm~\ref{algo:matrixF5-phomo} with the steps provided by Algorithm~\ref{algo:algosteps-matrixhomo} and with the GCD filter (Proposition~\ref{prop:sig-gcd}), without the GCD filter, and with the default strategy considering only the first system of weights.}
  \footnotesize
  \begin{tabu}{llllrrrrrrr}
    \toprule
    $\mathbf{W}$ & $r$ & $\mathbf{W}$-deg. & Algorithm & $\dmax$ & Deg & Steps &  Mtx/step & \pbox[r]{Max\\matrix\\size} & \pbox[r]{Reds \\ to 0} & Time (s)\\
    \midrule
    $\begin{psmallmatrix} 1 & 2 & 3 \\ 2 & 1 & 1 \end{psmallmatrix} $&
    2 & $(100,50)$ &
    GCD filter & 190 & 76
    & 10 & 1 & 400 & 0 & 1.29 \\
    &&& No filter & 190 & 76 
    & 91 & 147 & 400 & 0 & 186.89 \\
    &&& Default & \truncated{150} & \truncated{64} 
    & \truncated{51}  & 1 & \truncated{900k}  & \truncated{0}  & \truncated{754.36}  \\
    \midrule
    $\begin{psmallmatrix} 1 & 1 & 2 & 3 \\ 2 & 2 & 1 & 1 \end{psmallmatrix} $&
    2 & $(30,30)$ &
    GCD filter &  71 & 53 &
    42 & 5 & 23k &  0&  11.14 \\
    &&& No filter & 71  & 53
    & 42 & 65 & 23k & 0 & 33.83 \\
    &&& Default & \truncated{45} & \truncated{29} &
    \truncated{16} & 1 & \truncated{1100k}  & \truncated{0}  & \truncated{343.47}  \\
    \midrule
    $\begin{psmallmatrix} 1 & 1 & 2 & 3 \\ 2 & 2 & 1 & 1 \end{psmallmatrix} $&
    3 & $(30,30)$ &
    GCD filter &  85 & 44 &
    56 & 6 & 80k & 40 (9)&  33.94 \\
    &&& No filter & 85  & 44
    & 56 & 88 & 80k & 40 (9) & 88.91 \\
    &&& Default & \truncated{45} &  \truncated{34}
    & \truncated{16} & 1 & \truncated{1651k} & \truncated{0} & \truncated{517.19} \\
    \bottomrule
  \end{tabu}
  \label{tab:exp1}
\end{table*}


\end{document}